\documentclass[sts,preprint]{imsart}
\usepackage{geometry} 
\usepackage{graphicx}	
\usepackage{amssymb, amsmath, amsthm}
\usepackage{subfigure}
\usepackage{color}
\usepackage[table]{xcolor}
\usepackage{natbib}
\usepackage[colorlinks,citecolor=blue,urlcolor=blue]{hyperref}
\usepackage{comment}
\usepackage{todonotes}

\graphicspath{{Figures/}}

\newcommand{\cOmega}{\, \lrcorner \, \Omega}
\newcommand{\LS}{\ensuremath { H^{-1} \! \left( E \right) } }
\newcommand{\sLS}{\ensuremath { H^{-1} \left( E \right) } }
\newcommand{\MLS}{\ensuremath { \widetilde{H}^{-1} \! \left( E \right) } }
\newcommand{\sMLS}{\ensuremath { \widetilde{H}^{-1} \left( E \right) } }

\setattribute{journal}{name}{}

\theoremstyle{plain}
\newtheorem{theorem}{Theorem}
\newtheorem{lemma}[theorem]{Lemma}
\newtheorem{corollary}[theorem]{Corollary}

\theoremstyle{definition}

\theoremstyle{remark}

\theoremstyle{plain}

\begin{document}

\begin{frontmatter}

\title{Optimizing The Integrator Step Size for Hamiltonian Monte Carlo}
\runtitle{Optimizing The Integrator Step Size for Hamiltonian Monte Carlo}

\begin{aug}
  \author{Michael Betancourt%
  \ead[label=e1]{betanalpha@gmail.com}},
  \author{Simon Byrne},
  \and
  \author{Mark Girolami}

  \runauthor{Betancourt et al.}

  \address{Michael Betancourt is a Postdoctoral Research Associate at the University of Warwick, 
                 Coventry CV4 7AL, UK \printead{e1}.  Simon Byrne is an EPSRC Postdoctoral Research
                 Fellow at University College London, Gower Street, London, WC1E 6BT.  
                 Mark Girolami is an ESPRC Established Career Research Fellow at the University of Warwick, 
                 Coventry CV4 7AL, UK.}

\end{aug}

\begin{abstract}
Hamiltonian Monte Carlo can provide powerful inference in complex statistical problems,
but ultimately its performance is sensitive to various tuning parameters.  In this paper
we use the underlying geometry of Hamiltonian Monte Carlo to construct a universal
optimization criterion for tuning the step size of the symplectic integrator crucial to any 
implementation of the algorithm as well as diagnostics to monitor for any signs of
invalidity.  An immediate outcome of this result is that the suggested target average
acceptance probability of 0.651 can be relaxed to $0.6 \lesssim a \lesssim 0.9$ with 
larger values more robust in practice.
\end{abstract}


\end{frontmatter}

Hamiltonian Monte Carlo~\citep{DuaneEtAl:1987, Neal:2011, BetancourtEtAl:2014}
yields efficient inference that scales to high-dimensional problems by building up
Markov transitions from a smooth maps known as \textit{Hamiltonian flow}.  In 
practice this flow must simulated by a numerical integrator and, while any bias can 
be corrected by applying a Metropolis acceptance procedure, the performance of
the resulting algorithm is highly sensitive to the choice of integrator step size.

At present the integrator step size is typically set by applying the analysis in 
\cite{BeskosEtAl:2013}, which derived a lower bound for the cost of the algorithm 
as a function of the integrator step size when targeting an independently and
identically distributed target distribution and a second-order leapfrog integrator.  
In this paper we build upon that work, introducing a complementary upper bound
that admits a more robust optimization criterion for the integrator step size and 
leveraging the geometry inherent to the algorithm to extend the criterion to any 
target distribution and any symplectic integrator.

In order to build such an optimization criterion we first relate the computational 
cost of a Hamiltonian Monte Carlo transition to various expectations that depend 
on the integrator and the chosen step size.  We next show how to approximate 
these expectations and then use those results to construct a robust optimization
criterion by minimizing the cost.  Finally, we discuss how those approximations 
may fail in practice and how to compensate the optimization procedure to ensure
robust application.

\section{Bounding The Cost of a Hamiltonian Monte Carlo Transition} \label{sec:cost}

Hamiltonian Monte Carlo transitions are generated from a \textit{Hamiltonian},
\begin{equation*}
H \! \left(q, p \right) = T \! \left( q, p \right) + V \! \left( q \right),
\end{equation*}
where the \textit{kinetic energy}, $T \! \left( q, p \right)$, is specified by
the user subject to certain constraints \citep[Sec. 3.1.3]{BetancourtEtAl:2014}
and the \textit{potential energy}  is defined by a given target distribution,
$V \! \left( q \right) = - \log \varpi \! \left( q \right)$.  Beginning with an
initial position, $q$, each transition generates a joint state by randomly sampling 
the momenta, 
\begin{equation*}
p \sim \varpi_{q} \! \left( p | q \right) \propto \exp \! \left( - T \! \left(q, p \right) \right),
\end{equation*}
and then producing a new state by integrating Hamilton's equations,
\begin{align*}
\frac{dq}{dt} 
&= +\frac{ \partial H }{ \partial p } = + \frac{ \partial T }{ \partial p } \\
\frac{dp}{dt} 
&= -\frac{ \partial H }{ \partial q } = - \frac{ \partial T }{ \partial q } - \frac{ \partial V }{ \partial q }.
\end{align*}
for some time $\tau$.  

When Hamilton's equations are integrated exactly the joint state is distributed as
\begin{equation*}
\left( q, p \right) \sim \varpi_{H} \! \left( q, p \right) \propto \exp \! \left( - H \! \left(q, p \right) \right)
\end{equation*}
while the position is marginally distributed according to the desired target distribution,
\begin{equation*}
q \sim \varpi \! \left( q \right).
\end{equation*}
In practice, however, the Hamiltonian trajectory can be integrated only approximately 
and the stationary distribution of the position will be biased away from $\varpi$.

If we use a $k$th-order symmetric symplectic integrator with step size $\epsilon$ to simulate
the Hamiltonian trajectory, then we can exactly cancel this bias with a straightforward
Metropolis scheme.  First we compose the approximate Hamiltonian trajectory,
$\phi^{\tilde{H}}_{\epsilon, \tau} \! \left(q, p \right)$, with a momentum reversal operator,
$R : \left(q, p \right) \mapsto \left(q, -p \right)$, to generate a proposal,
\begin{equation*}
\Phi^{\tilde{H}}_{\epsilon, \tau} \! \left( q, p \right) 
= R \circ \phi^{\tilde{H}}_{\epsilon, \tau} \! \left(q, p \right).
\end{equation*}
This proposal is then accepted only with the probability,
\begin{equation*}
a \! \left(q, p \right) = 
\min \! \left(1, \exp \! \left( 
\Delta_{\epsilon} \! \left( q, p \right) \right)
\right),
\end{equation*}
where $\Delta_{\epsilon} \! \left( q, p \right)$ is the Hamiltonian error,
\begin{equation*}
\Delta_{\epsilon} \! \left( q, p \right) = 
- H \circ \Phi^{\tilde{H}}_{\epsilon, \tau} \! \left(q, p \right)
+ H \! \left(q, p \right).
\end{equation*}

The number of attempts required to produce an accepted proposal follows a
geometric distribution with the probability of success
\begin{equation*}
a \! \left( q \right) = \mathbb{E}_{\varpi_{q} } \! \left[ a \! \left(q, p \right) \right].
\end{equation*}
With the cost of generating a proposal just the cost of simulating at trajectory,
\begin{equation*}
C_{\mathrm{proposal}} = C_{\mathrm{iteration}} \frac{ \tau }{ \epsilon },
\end{equation*}
the expected cost of generating an accepted proposal is given by averaging the
expected number of rejections over the position space,
\begin{equation*}
\mathbb{E}_{q} \! \left[ C \right]
= \mathbb{E}_{\varpi} \! \left[ C_{\mathrm{proposal}} \, r \right] \\
= \frac{C_{\mathrm{iteration}}}{\epsilon} \mathbb{E}_{\varpi} \! \left[ \tau \, r \right].
\end{equation*}
If $\tau$ is chosen independently of position, then 
\begin{align*}
\mathbb{E}_{q} \! \left[ C \right]
&= C_{\mathrm{iteration}} \frac{ \tau }{ \epsilon } 
\cdot  \mathbb{E}_{\varpi} \left[ \frac{1}{a \! \left(q \right) } \right] \\
&= C_{\mathrm{iteration}} \frac{ \tau }{ \epsilon } 
\cdot \mathbb{E}_{\varpi} \left[ \frac{1}{\mathbb{E}_{\varpi_{q}} \left[ a \! \left( q, p \right) \right] } \right].
\end{align*}

Following \citet[Eq. 4.2]{BeskosEtAl:2013}, we apply Jensen's inequality to the
outer expectation to yield a lower bound on the expected cost. Jensen's
inequality, however, can also be applied on the inner expectation to give a
complementary upper bound,
\begin{equation} \label{expectation_bounds}
\frac{1}{ \mathbb{E}_{\varpi_{H}} \left[ a \! \left( q, p \right) \right] } 
\leq 
\mathbb{E}_{\varpi} \left[ \frac{1}{\mathbb{E}_{\varpi_{q}} \left[ a \! \left( q, p \right) \right] } \right] 
\leq 
\mathbb{E}_{\varpi_{H}} \left[ \frac{1}{a \! \left( q, p \right)} \right].
\end{equation}

These bounds are particularly advantageous because they reduce to expectations
of functions of the error in the Hamiltonian, $\Delta_{\epsilon} \! \left( q, p \right)$,
with respect to the joint distribution, $\varpi_{H}$.  These expectations
admit well-behaved approximations independent of the actual form of the
potential and kinetic energies and hence the particular details of the given
problem.

\section{Approximating Canonical Expectations}
\label{sec:approximating_canonical_expectations}

More formally, expectations with respect to the joint distribution, $\varpi_{H}$,
in Hamiltonian Monte Carlo are \textit{canonical expectations} and are
readily estimated in practice with symplectic integrators.  In this section we 
define canonical expectations and their relationship to Hamiltonian Monte Carlo, 
show how symplectic integrators approximate these canonical expectations and constrain 
the accuracy of these approximations in general, and then ultimately construct universal 
approximations to canonical expectations of certain functions of the Hamiltonian error.

This construction is necessarily technical and requires a strong familiarity with
differential geometry and the geometric foundations of Hamiltonian Monte Carlo
\citep{BetancourtEtAl:2014}.  We reserve the detailed proofs to Appendix
\ref{apx:main} and suggest that readers interested in only the final result
skip ahead to Section \ref{sec:approximate_bounds}.

\subsection{Canonical Expecations}

Hamiltonian systems, $\left(M, \omega, H \right)$, where $M$ is a smooth, $2n$-dimensional
manifold, $\omega$ a symplectic form, and $H : M \rightarrow \mathbb{R}$ a smooth
function, are particularly rich probabilistic systems.  In the following we review how probability
measures arise naturally on Hamiltonian systems, the implicit Hamiltonian system and 
corresponding measures driving Hamiltonian Monte Carlo, and how expectations with respect 
to such measures can be computed in theory.

\subsubsection{Canonical Distributions and Expectations}

On a Hamiltonian system the symplectic form, $\omega$, immediately defines a canonical 
volume form,
\begin{equation*}
\Omega = \wedge_{i = 1}^{n} \omega,
\end{equation*}
or, in canonical coordinates,
\begin{align*}
\Omega 
&= 
\mathrm{d} q^{1} \wedge \ldots \wedge \mathrm{d} q^{n} \wedge
\mathrm{d} p_{1} \wedge \ldots \wedge \mathrm{d} p_{n}.
\end{align*}
Provided that $\int_{M} e^{- \beta H} \Omega$ is finite for some $\beta \in \mathbb{R}$, we can 
also construct canonical probability measures,
\begin{equation*}
\varpi_{H} = \frac{ e^{ -\beta H } \Omega }{ \int_{M} e^{ -\beta H } \Omega },
\end{equation*}
known as \textit{canonical distributions}.  We refer to expectations of functions with
respect to canonical distributions as \textit{canonical expectations}.

The Hamiltonian foliates the manifold, $M$, into level sets,
\begin{equation*}
\LS = \left\{ z \in M \, | \, H \! \left( z \right) = E \right \},
\end{equation*}
and $\varpi$ naturally disintegrates into \textit{microcanonical distristributions},  $\pi_{H^{-1} ( E )}$,
that concentrate on these submanifolds,
\begin{align*}
\varpi_{H}
&= H^{*} \varpi_{E} \wedge \varpi_{ \sLS }
\\
&=
\frac{ e^{-\beta E} }{\int_{M} e^{-\beta H} \Omega } d \! \left( E \right) \mathrm{d} E
\wedge \frac{ \vec{v} \cOmega} { d \! \left( E \right) }.
\end{align*}
Here $\vec{v}$ is any transverse vector field satisfying $\mathrm{d} H \! \left( \vec{v} \right) = c$,
$\iota_{E} : H^{-1} \! \left( E \right) \hookrightarrow M$ is the inclusion of $H^{-1} \! \left( E \right)$ 
into $M$, and
\begin{equation*}
d \! \left( E \right) = 
\frac{ \int_{ H^{-1} \left(E\right) } \iota^{*}_{E} \! \left( \vec{v} \cOmega \right) }
{ \mathrm{d} H \! \left( v \right) }.
\end{equation*}
is the \textit{density of states}.  Without loss of generality we will always rescale $\vec{v}$ such
that $\mathrm{d} H \! \left( \vec{v} \right) = 1$.

Combined with the symplectic form, the Hamiltonian also generates a \textit{Hamiltonian flow},
\begin{equation*}
\phi^{H}_{t} : M \rightarrow M, \, t \in \mathbb{R},
\end{equation*}
under which both the symplectic volume form and Hamiltonian, and consequently
the canonical and microcanonical distributions, are invariant.

\subsubsection{Hamiltonian Monte Carlo and Canonical Expectations}

Because it preserves the canonical distribution, Hamiltonian flow can be used to construct an 
efficient Markov transition.  The only problem is that a given probability space, 
$\left( Q, \mathcal{B} \! \left( Q \right), \varpi \right)$, does not have the symplectic structure 
necessary be a Hamiltonian system.

Hamiltonian Monte Carlo leverages Hamiltonian flow by considering not the sample space,
$Q$, but rather it's cotangent bundle, $T^{*} Q$.  If $Q$ is a smooth and orientable
$n$-dimensional manifold then the cotangent bundle is itself a smooth, orientable 
$2n$-dimensional manifold with a canonical fiber bundle structure, 
$\pi : T^{*} Q \rightarrow Q$, and a canonical symplectic form, $\omega$.  

The target measure on $Q$, given in canonical coordinates by
\begin{equation*}
\varpi \propto e^{-V} \mathrm{d} q^{1} \wedge \ldots \wedge \mathrm{d} q^{n},
\end{equation*}
is lifted onto the cotangent bundle with the choice of a disintegration,
\begin{equation*}
\varpi_{q} \propto e^{-T} \mathrm{d} p_{1} \wedge \ldots \wedge \mathrm{d} p_{n}
+ \text{horizontal} \; n\text{-forms},
\end{equation*}
yielding the joint measure
\begin{align*}
\varpi_{H} 
&= \pi^{*} \varpi_{q} \wedge \varpi \\
&\propto e^{-T - V} 
\mathrm{d} p_{1} \wedge \ldots \wedge \mathrm{d} p_{n} \wedge 
\mathrm{d} q^{1} \wedge \ldots \wedge \mathrm{d} q^{n} \\
&\propto e^{- \left( T+ V \right) } \Omega.
\end{align*}

Taking,
\begin{equation*}
H = - \log \frac{ \mathrm{d} \varpi_{H} }{ \mathrm{d} \Omega } = T + V,
\end{equation*}
this lift defines a Hamiltonian system, $\left( T^{*} Q, \omega, H \right)$,
where the joint measure $\varpi_{H}$ is exactly the \textit{unit canonical measure} 
with $\beta = 1$.   In particular, expectations with respect to $\varpi_{H}$ are
all canonical expectations.

\subsubsection{Computing Canonical Expectations}

One of the important benefits of a Hamiltonian system is that the
canonical expectation of any smooth function, $f : M \rightarrow \mathbb{R}$, 
with respect to the canonical distribution
\begin{equation*}
\mathbb{E}_{\varpi_{H}} \! \left[ f \right] 
= \int_{M} f \, \varpi_{H}
= \frac{ \int_{M} f e^{-\beta H} \Omega }{ \int_{M} e^{-\beta H} \Omega },
\end{equation*}
can be computed by taking expectations with respect to the microcanonical distributions
on the level sets,
\begin{align*}
\mathbb{E}_{\varpi_{H}} \! \left[ f \right]
&= \int_{M} f \, \varpi_{H} \\
&= \int_{M} \varpi_{ \sLS } \wedge \varpi_{E} \\
&= 
\frac{1}{\int_{M} e^{-\beta H} \Omega } \int \mathrm{d} E \, d \! \left( E \right) e^{-\beta E}
\frac{\int_{ \sLS } \iota^{*}_{E} \! \left( f \, \vec{v} \cOmega \right)}
{ d \! \left( E \right) }.
\end{align*}

In particular the microcanonical expectations are readily computed using the Hamiltonian flow.
Birkhoff's ergodic theorem~\citep{Petersen:1989} states that given certain ergodicity conditions
the expectation of any function with respect to the microcanonical distribution is equal to its 
expectation along the Hamiltonian flow,
\begin{align*}
\frac{ \int_{\sLS} \iota^{*}_{E} \! \left( f \, \vec{v} \cOmega \right) }
{ d \! \left( E \right) }
&= 
\lim_{T \rightarrow \infty} \frac{1}{T} \int_{0}^{T} \mathrm{d} t \, \left( \phi^{H}_{t} \right)^{*} f
\\
&\equiv \left< f \right>_{ \sLS }.
\end{align*}

\subsection{Approximating Canonical Expectations with Symplectic Integrators}
\label{sec:accuracy_of_canonical_expectations}

The only problem with using Hamiltonian flow to compute expectations is that the Hamiltonian
flow itself requires the solution to a system of $2n$ first-order ordinary differential equations.  
For all but the simplest systems, analytical solution are unfeasible and we must instead 
resort to simulating the flow numerically.

Fortunately, there exist a family of numerical integrators that leverage the underlying symplectic 
geometry to conserve many of the properties of the exact 
flow~\citep{HairerEtAl:2006, LeimkuhlerEtAl:2004}.  These \textit{symplectic integrators} exactly 
preserve the symplectic volume form with only small variations in the Hamiltonian along the 
simulated flow.

In fact, symplectic integrators simulate some flow exactly, just not the flow corresponding to $H$.  
Using backwards error analysis one can show that a $k$-th order symmetric symplectic integrator
exactly simulates the flow for some \textit{modified Hamiltonian}, given by an even, asymptotic
expansion with respect to the integrator step size, $\epsilon$,
\begin{equation*}
\widetilde{H} = H 
+ \sum_{n=k / 2}^{N} \epsilon^{2 n} H_{\left( n \right) } 
+ \mathcal{O} \! \left( e^{- c / \epsilon } \right).
\end{equation*}
Because it is exponentially small in the step size, the asymptotic error is typically neglected
and the leading-order behavior of is given by
\begin{equation*}
\widetilde{H} = H + \epsilon^{k} G + \mathcal{O} ( \epsilon^{k + 2} ).
\end{equation*}

As in the exact case, the modified Hamiltonian foliates the manifold and we can define level sets,
\begin{equation*}
\MLS = \left\{ z \in M \, | \, \widetilde{H} \! \left( z \right) = E \right \},
\end{equation*}
a corresponding inclusion map,
\begin{equation*}
\tilde{\iota}_{E} : \widetilde{H}^{-1} \! \left( E \right) \hookrightarrow M,
\end{equation*}
and a corresponding transverse vector field,
\begin{equation*}
\mathrm{d} \widetilde{H} \! \left( \vec{u} \right) = 1.
\end{equation*}

Provided that the asymptotic error is indeed negligible and the symplectic integrator is
\textit{topologically stable}~\citep{McLachlanEtAl:2004}, the modified level sets will have
the same topology as the exact level sets.  In particular, when the exact foliation
defines a well-behaved disintegration into microcanonical distributions we can define 
a corresponding \textit{modified density of states},
\begin{equation*}
\tilde{d} \! \left( E \right) = 
\int_{ \sMLS } \tilde{\iota}^{*}_{E} \! \left( \vec{u} \cOmega \right),
\end{equation*}
\textit{modified microcanonical distribution},
\begin{equation*}
\varpi_{\sMLS} = \frac{ \vec{u} \cOmega }{ \tilde{d} \! \left( E \right) },
\end{equation*}
and \textit{modified canonical distribution},
\begin{align*}
\varpi_{\tilde{H}}
&= \tilde{H}^{*} \widetilde{\varpi}_{E} \wedge \varpi_{\sMLS}
\\
&=
\frac{ e^{-\beta E} }{\int_{M} e^{-\beta \widetilde{H}} \Omega }  \tilde{d} \! \left( E \right) \mathrm{d} E
\wedge \frac{ \vec{u} \cOmega} { \tilde{d} \! \left( E \right) }.
\end{align*}

Using the flow from a numerical integrator to compute averages yields expectations
with respect to these modified measures,
\begin{align*}
\mathbb{E}_{\varpi_{\widetilde{H}}} \! \left[ f \right]
&= \int_{M} f \, \varpi_{\widetilde{H}} \\
&= 
\frac{1}{\int_{M} e^{-\beta \widetilde{H}} \Omega } \int \mathrm{d} E \, \tilde{d} \! \left( E \right) e^{-E}
\frac{\int_{ \sMLS } \tilde{\iota}^{*}_{E} \! \left( f \, \vec{u} \cOmega \right)}
{ \tilde{d} \! \left( E \right) },
\end{align*}
where
\begin{align*}
\frac{ \int_{\sMLS} \tilde{\iota}^{*}_{E} \! \left( f \, \vec{u} \cOmega \right) }
{ \tilde{d} \! \left( E \right) }
&= 
\lim_{T \rightarrow \infty} \frac{1}{T} \int_{0}^{T} \mathrm{d} t \, \left( \phi^{\widetilde{H}}_{t} \right)^{*} f
\\
&\equiv \left< f \right>_{ \sMLS }.
\end{align*}

The ultimate utility of a symplectic integrator and its modified Hamiltonian system is in the accuracy 
of its expectations relative to the true canonical expectations.  Fortunately, the geometric structure
of symplectic integrators ensures that the approximation error of both microcanonical and canonical 
expectations computed with a symplectic integrator is well-controlled.

\begin{theorem} 
\label{thm:micro_canonical}
Let $\left( M, \omega, H\right)$ be a Hamiltonian system and consider
a $k$-th order symmetric symplectic integrator with the corresponding modified Hamiltonian
$\widetilde{H} = H + \epsilon^{k} \, G + \mathcal{O} ( \epsilon^{k + 2} )$.  
If the integrator is topologically stable and the asymptotic error is negligible, 
then the difference in the microcanonical and modified microcanonical expectations for 
any smooth function, $f : M \rightarrow \mathbb{R}$, is given by
\begin{align*}
\left< f \right>_{\sLS}- \left< f \right>_{\sMLS} 
=&
+ \epsilon^{k}
\left< 
f \, \mathrm{d} G \! \left( \vec{v} \right) + G \, \mathrm{d} f \! \left( \vec{v} \right) 
+ f \, G \left( \frac{\partial v^{i} }{ \partial q^{i} } + \frac{\partial v_{i} }{ \partial p_{i} } \right) 
\right>_{ \sMLS } 
\\
& - \epsilon^{k}
\Big< f \Big>_{\sMLS} 
\left< 
\mathrm{d} G \! \left( \vec{v} \right)
+ G \left( \frac{\partial v^{i} }{ \partial q^{i} } + \frac{\partial v_{i} }{ \partial p_{i} } \right) 
\right>_{ \sMLS }
+ \mathcal{O} ( \epsilon^{k + 2} ),
\end{align*}
where $\vec{v}$ is any transverse vector field satisfying $\mathrm{d} H \! \left( \vec{v} \right) = 1$.

\end{theorem}

Note that, when \LS and \MLS intersect at some initial point, this reduces to the calculation in 
\cite{ArizumiEtAl:2012}.  Moreover, to leading-order we can replace to the expectations 
over $\MLS$ on the RHS with expectations over $\LS$ to give
\begin{align*}
\left< f \right>_{\sLS} - \left< f \right>_{ \sMLS } 
=&
+ \epsilon^{k}
\left< 
f \, \mathrm{d} G \! \left( \vec{v} \right) + G \, \mathrm{d} f \! \left( \vec{v} \right) 
+ f \, G \left( \frac{\partial v^{i} }{ \partial q^{i} } + \frac{\partial v_{i} }{ \partial p_{i} } \right) 
\right>_{ \sLS } 
\\
& - \epsilon^{k}
\left< f \right>_{ \sLS } 
\left< 
\mathrm{d} G \! \left( \vec{v} \right)
+ G \left( \frac{\partial v^{i} }{ \partial q^{i} } + \frac{\partial v_{i} }{ \partial p_{i} } \right) 
\right>_{ \sLS }
+ \mathcal{O} ( \epsilon^{k + 2} ).
\end{align*}
This is convenient for numerical experiments when the canonical expectations
can be computed analytically.

Given the decomposition of the canonical distributions over level sets, the result for the 
accuracy of microcanonical expectations immediately carries over to a canonical
expectations.

\begin{theorem}
\label{thm:canonical}
Let $\left( M, \omega, H\right)$ be a Hamiltonian system and consider
a $k$-th order symmetric symplectic integrator with the corresponding modified Hamiltonian
$\widetilde{H} = H + \epsilon^{k} \, G + \mathcal{O} ( \epsilon^{k + 2} )$.  
If the integrator is topologically stable and the asymptotic error is negligible, 
then the difference in the canonical and modified canonical expectations for any smooth
function, $f : M \rightarrow \mathbb{R}$, is given by
\begin{align*}
\mathbb{E}_{\varpi_{H}} \! \left[ f \right] - \mathbb{E}_{\varpi_{\widetilde{H}}} \! \left[ f \right]
=&
+ \epsilon^{k} \,
\mathbb{E}_{\varpi_{H}} \! \left[
f \, \mathrm{d} G \! \left( \vec{v} \right) + G \, \mathrm{d} f \! \left( \vec{v} \right) 
+ f \, G \left( \frac{\partial v^{i} }{ \partial q^{i} } + \frac{\partial v_{i} }{ \partial p_{i} } \right) 
\right] 
\\
& - \epsilon^{k} \,
\mathbb{E}_{\varpi_{H}} \! \left[ f \right] \,
\mathbb{E}_{\varpi_{H}} \! \left[
\mathrm{d} G \! \left( \vec{v} \right)
+ G \left( \frac{\partial v^{i} }{ \partial q^{i} } + \frac{\partial v_{i} }{ \partial p_{i} } \right) \right] 
+ \mathcal{O} ( \epsilon^{k + 2} ),
\end{align*}
where $\vec{v}$ is any transverse vector field satisfying $\mathrm{d} H \! \left( \vec{v} \right) = 1$.

\end{theorem}

\subsection{Approximating Canonical Expectations of the Hamiltonian Error}
\label{sec:accuracy_of_canonical_expectations_of_hamiltonian_error}

The expectations necessary for bounding the cost of a basic Hamiltonian Monte Carlo
transition are not just any canonical expectations but canonical expectations of functions 
of the Hamiltonian error,
\begin{equation*}
\Delta_{\epsilon} = H - H \circ \Phi^{\widetilde{H}}_{\epsilon, \tau},
\end{equation*}
where again
\begin{equation*}
\Phi^{\widetilde{H}}_{\epsilon, \tau} = R \circ \phi^{\widetilde{H}}_{\epsilon, \tau}
\end{equation*} 
is the Metropolis proposal.  By constraining how the moments and then the
cumulants of the Hamiltonian error scale with the symplectic integrator step size we
can construct universal approximations to these particular expectations.

\subsubsection{Moments of the Hamiltonian Error}

\begin{lemma}
\label{lem:moment_scaling}
Let $\left( M, \omega, H\right)$ be a Hamiltonian system and consider
a $k$-th order symmetric symplectic integrator with the corresponding modified Hamiltonian
$\widetilde{H} = H + \epsilon^{k} \, G + \mathcal{O} ( \epsilon^{k + 2} )$
and Metropolis proposal $\Phi^{\widetilde{H}}_{\epsilon, \tau}$. 
If the integrator is topologically stable and the asymptotic error is negligible, 
then the moments of the Hamiltonian error with respect to the unit canonical distribution
scale as
\begin{equation*}
\mathbb{E}_{\varpi_{H}} \! \left[ \left( \Delta_{\epsilon} \right)^{n} \right] \propto
\left\{
\begin{array}{rr}
\epsilon^{ k ( n + 1 )}, & n \, \mathrm{odd} \;\, \\
\epsilon^{k n} \;\,, & n \, \mathrm{even}
\end{array} 
\right. .
\end{equation*}
to leading-order in $\epsilon$.

\end{lemma}

The mean of the Hamiltonian error is particularly interesting because it can be
computed analytically (Appendix \ref{apx:ave_error}).  In the case of a Gaussian
target distribution, a Euclidean kinetic energy, and a second-order leapfrog
integrator we have the Hamiltonian,
\begin{equation*}
H = \frac{1}{2} p^{2} + \frac{1}{2} q^{2},
\end{equation*} 
the sub-leading contribution to the modified Hamiltonian,
\begin{equation*}
G = 
\frac{1}{24} \left( 2 q^{2} - p^{2} \right),
\end{equation*}
and eventually the average error
\begin{align*}
\mathbb{E}_{\varpi_{H}} \left[ \Delta_{\epsilon} \right] 
&=
\frac{1}{64} \epsilon^{4} \left(1 - \cos 2 \tau \right)
+ \mathcal{O} \! \left( \epsilon^{6} \right),
\end{align*}
in agreement with numerical experiments (Figure \ref{fig:delta}).

\begin{figure}
\centering
\includegraphics[width=4in]{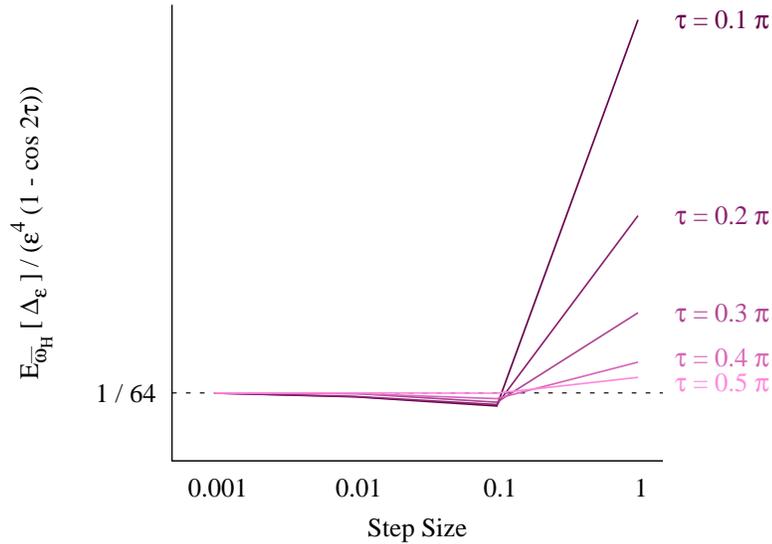}
\caption{For the Gaussian target distribution, numerical estimates of 
$\mathbb{E}_{\varpi_{H}} \left[ \Delta_{\epsilon} \right]$ are consistent with the 
geometric results for a wide range of integration times, $\tau$, and step sizes, 
$\epsilon$, at least until higher-order contributions become significant at large
step sizes.  The estimates were performed using an analytic form of the error 
such that discretization of the integration time is not a factor at larger step sizes.}
\label{fig:delta}
\end{figure}

\subsubsection{Cumulants of the Hamiltonian Error}

\begin{lemma}
\label{lem:cumulant_scaling}
Let $\left( M, \omega, H\right)$ be a Hamiltonian system and consider
a $k$-th order symmetric symplectic integrator with the corresponding modified Hamiltonian
$\widetilde{H} = H + \epsilon^{k} \, G + \mathcal{O} ( \epsilon^{k + 2} )$
and Metropolis proposal $\Phi^{\widetilde{H}}_{\epsilon, \tau}$.  
If the integrator is topologically stable and the asymptotic error is negligible, 
then the cumulants of the Hamiltonian error with respect to the unit canonical distribution scale 
as
\begin{equation*}
\kappa_{n} \! \left( \Delta_{\epsilon} \right) \propto
\left\{
\begin{array}{rr}
\epsilon^{k (n + 1) }, & n \, \mathrm{odd} \;\, \\
\epsilon^{k n} \;\,, & n \, \mathrm{even}
\end{array} 
\right. .
\end{equation*}
to leading-order in $\epsilon$.
\end{lemma}

\begin{lemma} 
\label{lem:global_constraint}

Let $\left( M, \omega, H \right)$ be a Hamiltonian system and consider
a $k$-th order symmetric symplectic integrator with the corresponding modified Hamiltonian
$\widetilde{H} = H + \epsilon^{k} \, G + \mathcal{O} ( \epsilon^{k + 2} )$
and Metropolis proposal $\Phi^{\widetilde{H}}_{\epsilon, \tau}$. 
If the integrator is topologically stable and the asymptotic error is negligible, 
then the cumulant generating function of the Hamiltonian error vanishes
\begin{equation*}
\log \mathbb{E}_{\varpi_{H}} \! \left[ e^{ \Delta_{\epsilon}} \right]  = 0.
\end{equation*}
for the unit canonical distribution,
\begin{equation*}
\varpi_{H} = \frac{ e^{-H} \Omega }{ \int_{M} e^{-H} \Omega }.
\end{equation*}

\end{lemma}

When coupled with Jensen's inequality, Lemma \ref{lem:global_constraint}
implies that $\mathbb{E}_{\varpi_{H}} \! \left[ \Delta_{\epsilon} \right] \le 0$ with equality 
holding only for an exact integrator when $\Delta_{\epsilon}$ is identically zero.  This shows that 
a symplectic integrator will always introduce an error in expectation and the average Metropolis 
acceptance probability will always be smaller than unity.

\begin{corollary}
\label{cor:leading_cumulants}
Let $\left( M, \omega, H \right)$ be a Hamiltonian system and consider
a $k$-th order symmetric symplectic integrator with the corresponding modified Hamiltonian
$\widetilde{H} = H + \epsilon^{k} \, G + \mathcal{O} ( \epsilon^{2k} )$
and Metropolis proposal $\Phi^{\widetilde{H}}_{\epsilon, \tau}$. 
If the integrator is topologically stable and the asymptotic error is negligible, 
then to leading-order in $\epsilon$ the first two cumulants of the unit canonical distribution 
satisfy
\begin{equation*}
\kappa_{1} = - \frac{1}{2} \kappa_{2}.
\end{equation*}
\end{corollary}

\begin{proof}

The cumulant generating function gives
\begin{equation*}
\log \mathbb{E}_{\varpi_{H}} \left[ e^{ t \Delta_{\epsilon}} \right]
=
\sum_{n = 0}^{\infty} \frac{t^{n} \kappa_{n} }{n!},
\end{equation*}
or taking $t = 1$ and appealing to Lemma \ref{lem:global_constraint},
\begin{equation*}
0 = \log \mathbb{E}_{\varpi_{H}} \left[ e^{ \Delta_{\epsilon}} \right]
=
\sum_{n = 0}^{\infty} \frac{ \kappa_{n} }{n!},
\end{equation*}
But from Lemma \ref{lem:cumulant_scaling} we know that to
leading-order only the first two cumulants contribute,
\begin{equation*}
0 = \kappa_{1} + \frac{1}{2} \kappa_{2} + \mathcal{O} ( \epsilon^{2k} ),
\end{equation*}
or
\begin{equation*}
\kappa_{1} = - \frac{1}{2} \kappa_{2} + \mathcal{O} ( \epsilon^{2k} ),
\end{equation*}
as desired.

\end{proof}

Explicitly introducing the scaling from Lemma \ref{lem:cumulant_scaling} gives
\begin{equation*}
\kappa_{2} \approx - 2 \kappa_{1} = \alpha \, \epsilon^{2k}, \, \alpha \in \mathbb{R}.
\end{equation*}
At this point we can note that if the target distribution composes into $d$ independently
and identically distributions components then the cumulants scale as
\begin{equation*}
\kappa_{n} \! \left( \Delta_{\epsilon} \right) \propto
\left\{
\begin{array}{rr}
d \, \epsilon^{ k (n + 1)}, & n \, \mathrm{odd} \;\, \\
d \, \epsilon^{k n} \;\,, & n \, \mathrm{even}
\end{array} 
\right. .
\end{equation*}
If we scale the step size as $\epsilon^{2k} = \epsilon_{0}^{2k} / d$ then the cumulants
scaling becomes
\begin{equation*}
\kappa_{n} \! \left( \Delta_{\epsilon} \right) \propto
\left\{
\begin{array}{rr}
d^{ \frac{1 - n}{2}} \, \epsilon_{0}^{k (n + 1)}, & n \, \mathrm{odd} \;\, \\
d^{ \frac{2 - n}{2}} \, \epsilon_{0}^{kn} \;\,, & n \, \mathrm{even}
\end{array} 
\right. .
\end{equation*}

Consequently, in the infinite limit $d \rightarrow \infty$ all of the cumulants
beyond second-order vanish, $\varpi_{\Delta_{\epsilon}}$ converges to a
$\mathcal{N} \! \left( - \frac{1}{2} \alpha \epsilon^{2k}, \alpha \epsilon^{2k} \right)$ 
in distribution, and the desired expectations simply to Gaussian integrals.  
Extending this argument to independently but not necessarily identically 
distributed distributions corresponds to the results in \cite{BeskosEtAl:2013} 
generalized to any symplectic integrator.

Fortunately, even outside of the limit of infinite independently distributed
distributions the expectations are remarkably well-behaved.

\subsubsection{Expectations of the Hamiltonian Error}

Together, these Lemmas imply that canonical expectations of any smooth function 
of the Hamiltonian error, as well as the Metropolis acceptance probability with its 
single cusp, are well-approximated by straightforward Gaussian integrals.

\begin{theorem}
\label{thm:smooth_expectations}

Let $\left( M, \omega, H\right)$ be a Hamiltonian system and consider
a $k$-th order symmetric symplectic integrator with the corresponding modified Hamiltonian
$\widetilde{H} = H + \epsilon^{k} \, G + \mathcal{O} ( \epsilon^{k + 2} )$
and Metropolis proposal $\Phi^{\widetilde{H}}_{\epsilon, \tau}$.  If the integrator 
is topologically stable and the asymptotic error is negligible, then the expectation of 
any smooth function of the Hamiltonian error is given by
\begin{align*}
\mathbb{E}_{\varpi_{H} } \! \left[ f \! \left( \Delta_{\epsilon} \right) \right] 
&=
\int^{\infty}_{-\infty} \mathrm{d} \Delta_{\epsilon} \, 
\mathcal{N} \! \left( \Delta_{\epsilon} | 
- \frac{1}{2} \alpha \, \epsilon^{2k}, \alpha \, \epsilon^{2k} \right) 
f \! \left( \Delta_{\epsilon} \right)
+ \mathcal{O} ( \epsilon^{2k + 4} ),
\end{align*}
for some $\alpha \in \mathbb{R}$.

\end{theorem}

\begin{theorem}
\label{thm:acceptance_expectation}

Let $\left( M, \omega, H\right)$ be a Hamiltonian system and consider
a $k$-th order symmetric symplectic integrator with the corresponding modified Hamiltonian
$\widetilde{H} = H + \epsilon^{k} \, G + \mathcal{O} ( \epsilon^{k + 2} )$
and Metropolis proposal $\Phi^{\widetilde{H}}_{\epsilon, \tau}$.  If the integrator 
is topologically stable and the asymptotic error is negligible, then the expectation of 
the Metropolis acceptance probability is given by
\begin{align*}
\mathbb{E}_{\varpi_{H}} \! \left[ a \! \left( \Delta_{\epsilon} \right) \right] 
=&
\int^{\infty}_{-\infty} \mathrm{d} \Delta_{\epsilon} \, 
\mathcal{N} \! \left( \Delta_{\epsilon} | 
- \frac{1}{2} \alpha \, \epsilon^{2 k}, \alpha \, \epsilon^{2 k} \right) 
a \! \left( \Delta_{\epsilon} \right)
+ \mathcal{O} ( \epsilon^{k} ),
\end{align*}
for some $\alpha \in \mathbb{R}$.

\end{theorem}

\section{Approximating Bounds and the Step Size Optimization Criterion}
\label{sec:approximate_bounds}

The approximation expectations in Theorem \ref{thm:smooth_expectations}
and \ref{thm:acceptance_expectation} immediately admit universal,
approximation bounds on the cost of a basic Hamiltonian Monte Carlo
transition, and minimizing these bounds provides a correspondingly
universal strategy for tuning the integrator step size.

\begin{theorem}
\label{lem:cost}
Provided that the symmetric symplectic integrator is topologically stable and the 
asymptotic error is negligible, the cost of a Hamiltonian Monte Carlo implementation 
with Metropolis proposal $\Phi^{\widetilde{H}}_{\epsilon, \tau}$ is bounded by functions 
of the average Metropolis acceptance probability, 
$a = \mathbb{E}_{\varpi_{H}} \left[ a \! \left( q, p \right) \right]$,
\begin{equation*}
\frac{1}{ \epsilon \! \left( a \right) a } + \mathcal{O} ( \epsilon^{k} ) \le
\frac{1}{\tau} \frac{ \left< C \right> }{ C_{\mathrm{iteration}} } 
\le 
\frac{\frac{a}{2} + \Phi \! \left( - 3 \Phi^{-1} \! \left( a / 2  \right) \right) 
\exp \! \left( 4 \left( \Phi^{-1} \! \left( a /2  \right) \right)^{2} \right) }
{ \epsilon \! \left( a \right) }
+ \mathcal{O} ( \epsilon^{k} ),
\end{equation*}
where
\begin{equation*}
\epsilon \! \left( a \right) = 
\left[ \sqrt{ \frac{ 2 }{ \alpha } } \Phi^{-1} \! \left( 1 - \frac{a}{2} \right) \right]^{1/k}
\end{equation*}
for some $\alpha \in \mathbb{R}$.

\end{theorem}

\begin{proof}

Approximations for the both the lower and upper bounds in \eqref{expectation_bounds} 
are given immediately by carrying out the Gaussian integrals analytically 
\citep{RobertsEtAl:1997},
\begin{align*}
\mathbb{E}_{\varpi_{H}} \left[ a \! \left( q, p \right) \right]
&= 
2 \, \Phi \! \left( - \sqrt{ \frac{\alpha}{2} } \epsilon^{k} \right)
+ \mathcal{O} ( \epsilon^{k} )
\\
\\
\mathbb{E}_{\varpi_{H}} \left[ 1 / a \! \left( q, p \right) \right] 
&= 
\mathbb{E}_{\varpi_{H}} \left[ 1 + e^{- \Delta_{\epsilon}} \right] 
- \mathbb{E}_{\varpi_{H}} \left[ a \! \left( -\Delta_{\epsilon} \right) \right]
\\
&= 
\Phi \! \left( - \sqrt{ \frac{\alpha}{2} } \epsilon^{k} \right) 
+ \Phi \! \left( 3 \sqrt{ \frac{\alpha}{2} } \epsilon^{k} \right) 
e^{2 \, \alpha \, \epsilon^{2 k}}
+ \mathcal{O} ( \epsilon^{k} ).
\end{align*} 

Following previous work \citep{RobertsEtAl:1997, BeskosEtAl:2013} we now consider 
the cost as a function of not the step size but rather the average acceptance probability,
\begin{equation} 
\label{eqn:acceptVsStepsize}
a \equiv \mathbb{E}_{\varpi_{H}} \left[ a \! \left( q, p \right) \right] 
= 2 \, \Phi \! \left( - \sqrt{ \frac{\alpha}{2} } \epsilon^{k} \right) + \mathcal{O} ( \epsilon^{k} ),
\end{equation}
Solving for the step size yields
\begin{equation*}
\epsilon \! \left( a \right) = 
\left[ \sqrt{ \frac{ 2 }{ \alpha } } \Phi^{-1} \! \left( 1 - \frac{a}{2} \right) \right]^{1/k},
\end{equation*}
and subsequently substituting into the bounds gives
\begin{equation*}
\frac{1}{ a } + \mathcal{O} ( \epsilon^{k} )
\leq 
\mathbb{E}_{\varpi} \left[ \frac{1}{\mathbb{E}_{\varpi_{q}} \left[ a \! \left( q, p \right) \right] } \right] 
\leq 
\frac{a}{2} + \Phi \! \left( - 3 \Phi^{-1} \! \left( a / 2  \right) \right) 
\exp \! \left( 4 \left( \Phi^{-1} \! \left( a /2  \right) \right)^{2} \right)
+ \mathcal{O} ( \epsilon^{k} ),
\end{equation*}
or in terms of the cost,
\begin{equation*}
\frac{1}{ \epsilon \! \left( a \right) a } + \mathcal{O} ( \epsilon^{k} ) \le
\frac{1}{\tau} \frac{ \left< C \right> }{ C_{\mathrm{iteration}} } 
\le 
\frac{\frac{a}{2} + \Phi \! \left( - 3 \Phi^{-1} \! \left( a / 2  \right) \right) 
\exp \! \left( 4 \left( \Phi^{-1} \! \left( a /2  \right) \right)^{2} \right) }
{ \epsilon \! \left( a \right) }
+ \mathcal{O} ( \epsilon^{k} ),
\end{equation*}
as desired.

\end{proof}

Provided that the approximations hold, we can determine an optimal average
acceptance probability, and hence a criterion for tuning the integrator step size,
by minimizing these bounds.  The dependence on the particular problem is 
isolated to the $\alpha^{k / 2}$ scaling common to both bounds and hence does
not effect the resulting optimimum; consequently the optimal average acceptance 
probability is the same for all choices of the potential and kinetic energies and 
hence defines a universal tuning strategy.

For example, with a second-order symplectic integrator the lower bound is
minimized at $a \! \left( \epsilon \right) = 0.651$ while the upper bound is minimized
at $a \! \left( \epsilon \right) = 0.801$.  Because the bounds are relatively
flat between these two optima any target acceptance probability between
$0.6 \lesssim a \! \left( \epsilon \right) \lesssim 0.9$ essentially yields equivalent
results (Figure \ref{fig:bounds}).

\begin{figure}
\centering
\includegraphics[width=3in]{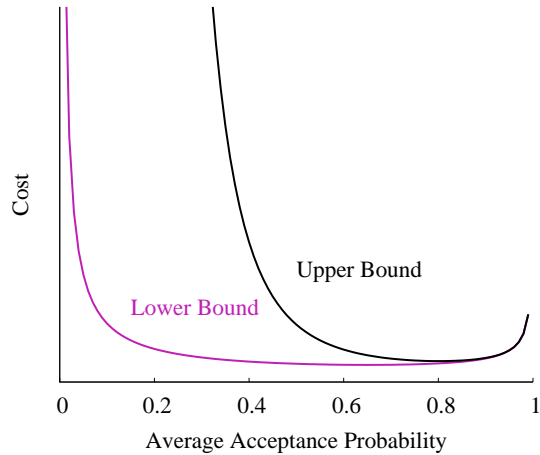}
\caption{The cost of a Hamiltonian Monte Carlo implementation using the Metropolis
proposal $\Phi^{\widetilde{H}}_{\epsilon, \tau}$ and a second-order symplectic integrator 
is squeezed between two approximate bounds that, up to a common scale, are universal 
for any target distribution and provide a general criteria for the optimal average acceptance 
probability or, equivalently, the integrator step size.  Maximizing the approximate lower bound 
yields $a = 0.651$ and minimizing the approximate upper bound suggests $a = 0.801$; given 
that the approximate bounds are relatively flat, however, any value in between yields 
near-optimal results.}
\label{fig:bounds} 
\end{figure}

\section{Limitations of the Step Size Optimization Criterion}

When applying this optimization criterion we have to be careful to account for both the
fundamental limitations in its construction and the possibility that the underlying
assumptions may fail.

For example, although the cost function is applicable to both a constant integration time 
and an integration time chosen uniformly over some static distribution it is not applicable
to an integration time that varies with the initial position, as would be necessary for a 
dynamically optimized integration time \citep{Betancourt:2013a}.  Technically this 
precludes implementations of Hamiltonian Monte Carlo like the No-U-Turn sampler
\citep{HoffmanEtAl:2014}, although in practice it has performed well as the default 
tuning mechanism for Stan \citep{Stan:2014}.

Similarly, the optimization criterion is only as good as the approximate bounds
from which is it constructed.  One source of error in these bounds are the
high-order contributions beyond the Gaussian integral, although empirically
these appear to be small for simple models (Figure \ref{fig:gaussNumerics}).
Because more complex models typically require smaller step sizes to achieve
the same average acceptance probability, the higher-order contributions should
continue to be negligible.

\begin{figure}
\centering
\subfigure[]{\includegraphics[width=2.75in]{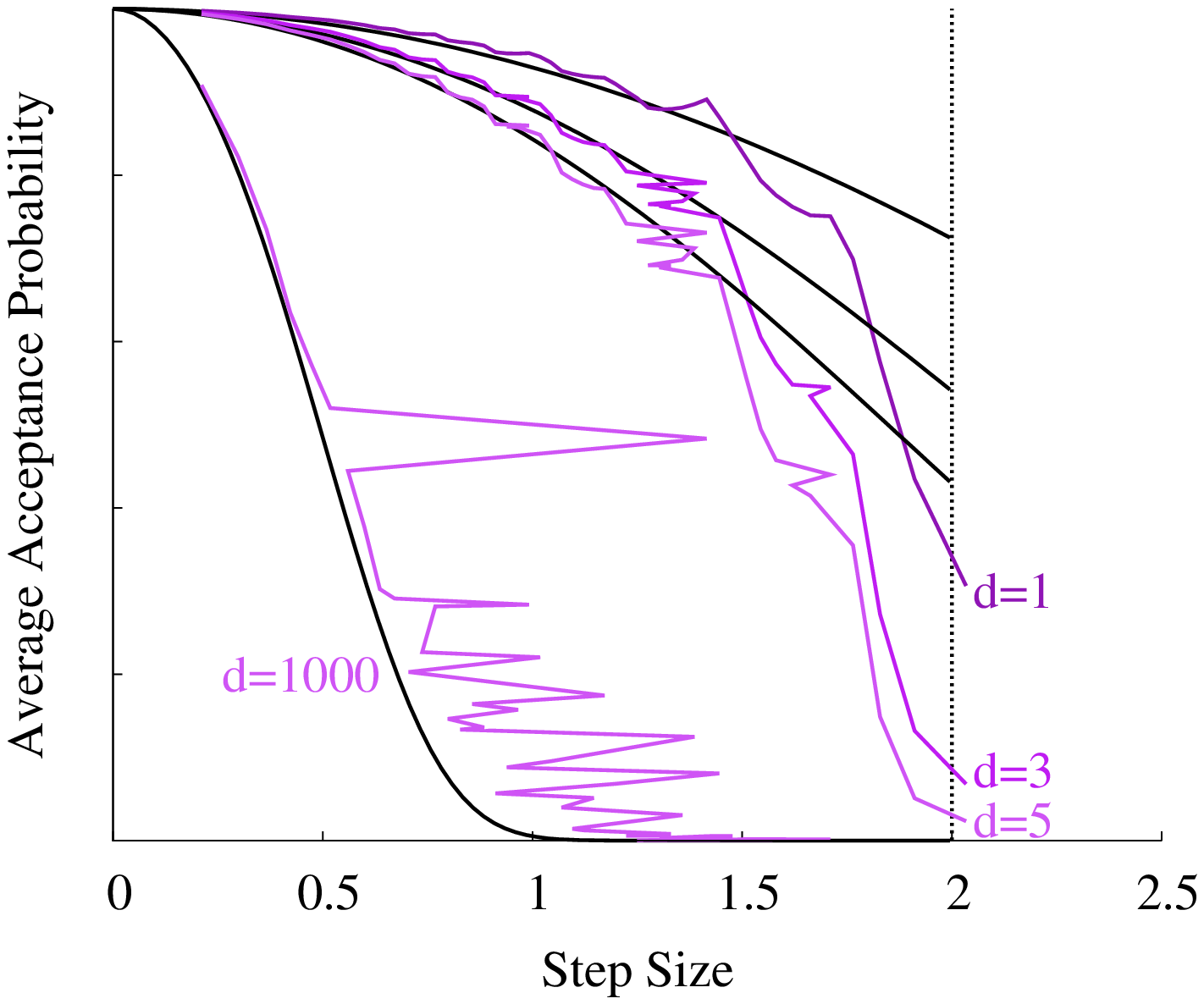}}
\subfigure[]{\includegraphics[width=2.75in]{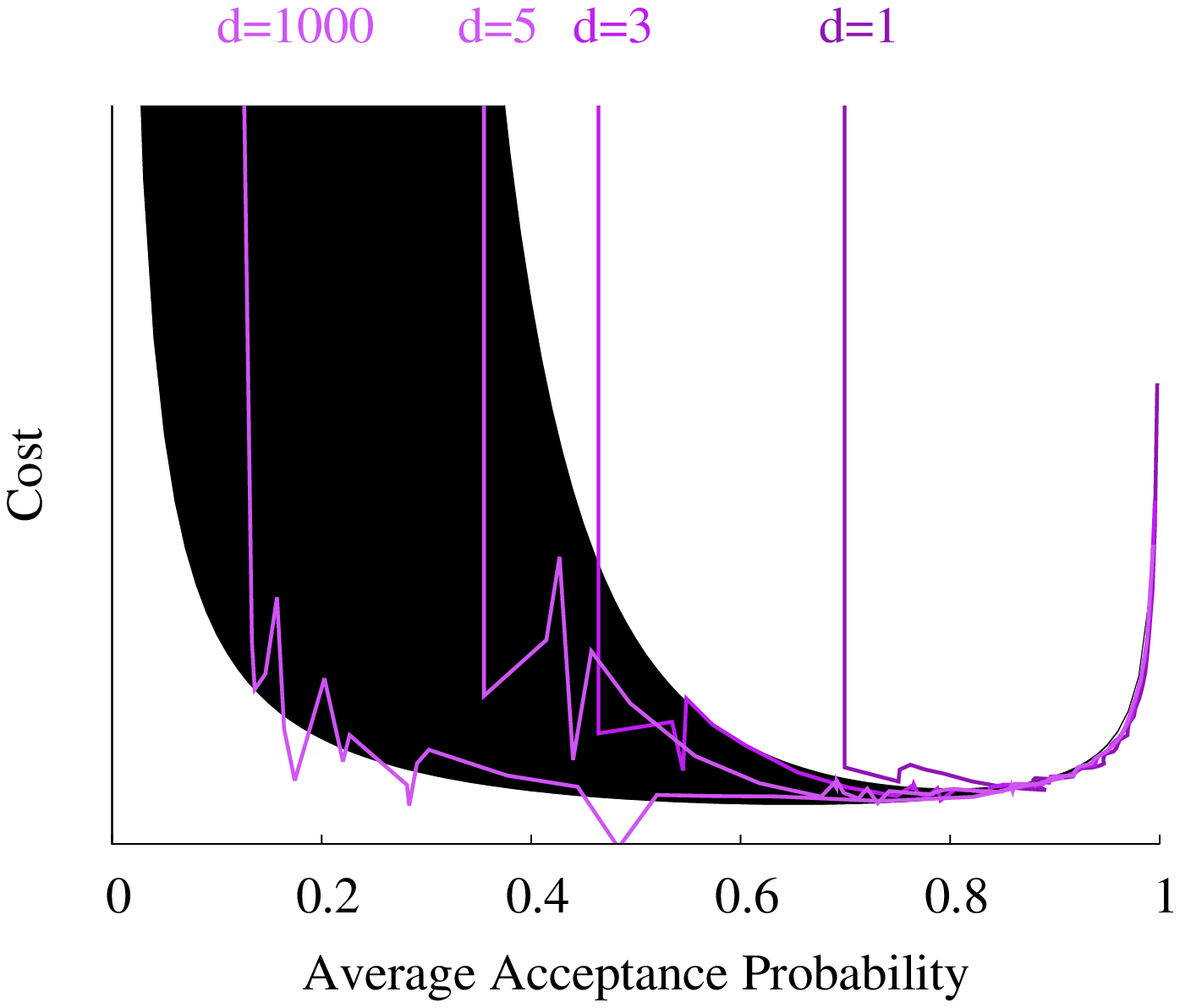}}
\caption{In experiments of Euclidean Hamiltonian Monte Carlo, implemented 
with a unit metric and second-order leapfrog integrator and targeting a product
of independent Gaussian distributions, both (a) the relationship between integrator
step size and average acceptance probability and (b) the relationship between
the average acceptance probability and the cost are in excellent agreement
with the approximations.  In this case any higher-order corrections are
negligible, at least until the step size approaches $\epsilon = 2$ at which point
the integrator becomes topologically unstable. }
\label{fig:gaussNumerics}
\end{figure}

A more serious concern with the approximations is not in the neglected
higher-order contributions but rather the assumption of topological stability and
vanishing asymptotic error.  In complex models the step sizes necessary for
these conditions to hold can be much smaller than the step size motivated by
the optimization strategy.  Fortunately, when these conditions do not hold the
integrator becomes unstable, manifesting almost immediately in numerical divergences
that pull the state towards infinity and are readily incorporated into user-facing diagnostics.  
Consequently our initial optimization strategy can be made robust by monitoring these
diagnostics and increasing the target average acceptance probability until no divergences 
occur (Figure \ref{fig:funnelScan}). This more robust strategy has proven especially
effective for hierarchical models that are particularly sensitive to this pathology 
\citep{BetancourtEtAl:2015}.

\begin{figure}
\centering
\subfigure[]{\includegraphics[width=2.75in]{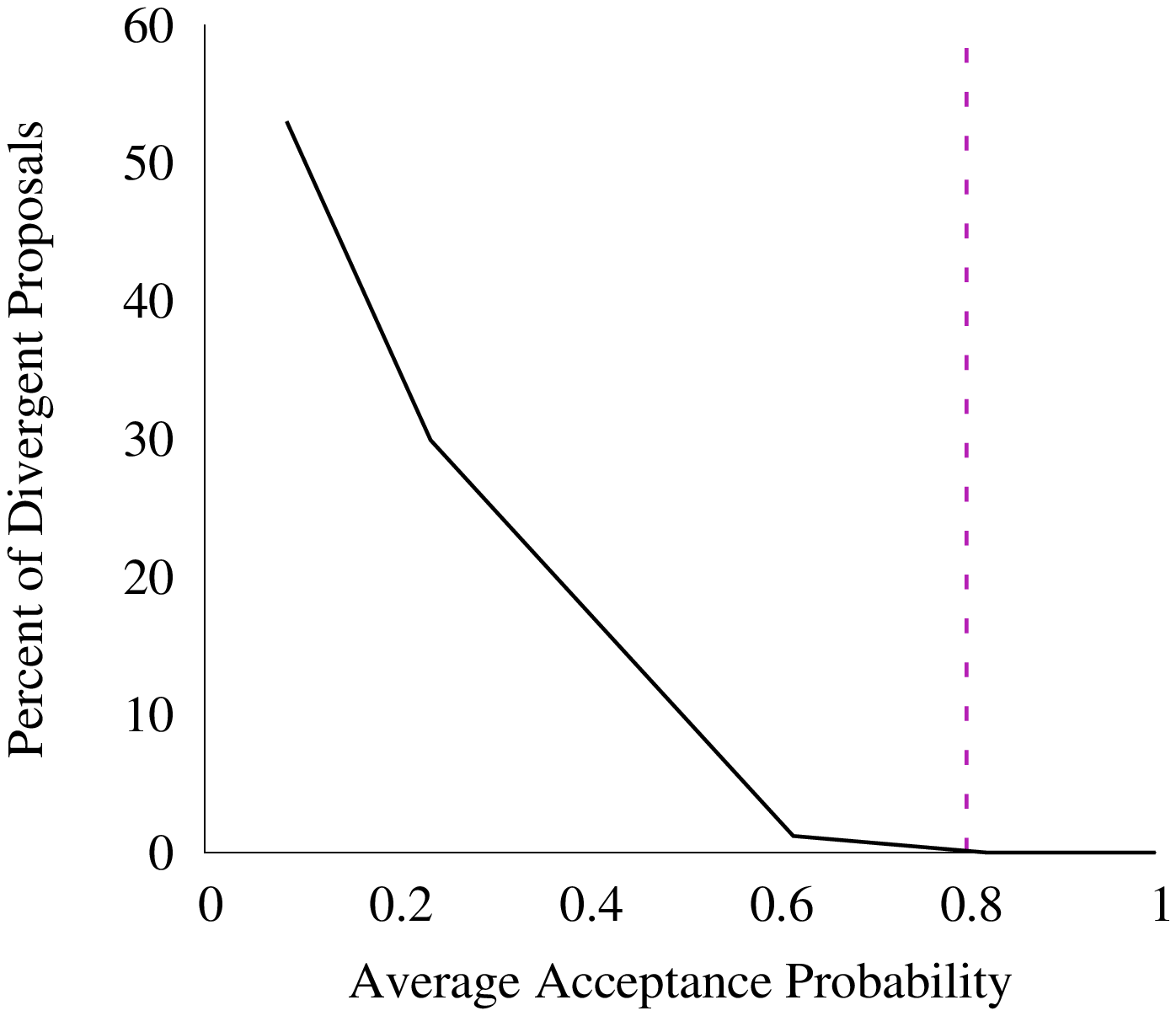}}
\subfigure[]{\includegraphics[width=2.75in]{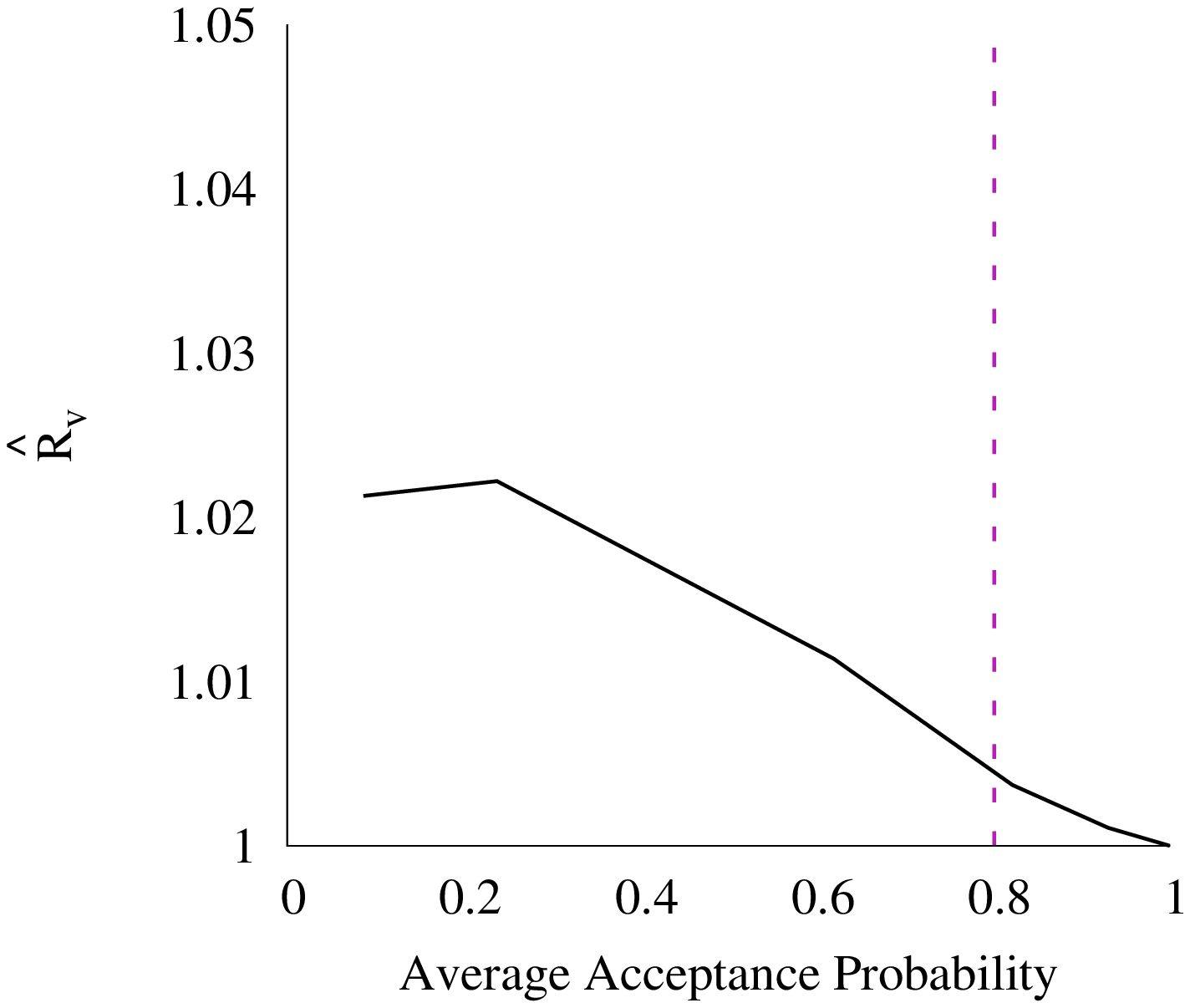}}
\caption{In complex problems careful monitoring of the step size optimization is critical 
towards ensuring a robust result.  For example, in the (50 + 1)-dimensional funnel distribution
\cite{Neal:2003} (a) numerical divergences begin to appear within the window of 
nominally-optimal acceptance probabilities which manifest as (c) an increasing potential 
scale reduction factor~\citep{GelmanEtAl:1992, StanMan:2014}, $\hat{R}_{v}$, that 
incorporates exact samples from the posterior.  As the model complexity increases these
inconsistencies become more severe and can result in biased Markov chains unless the 
optimization criterion is relaxed.
}
\label{fig:funnelScan}
\end{figure}

\section{Conclusions and Future Work}

By appealing to the underlying geometry of Hamiltonian Monte Carlo we have constructed
a robust, universal scheme for the optimal tuning of the integrator step size, for
for simple implementations of Hamiltonian Monte Carlo that use an approximate
Hamiltonian flow to construct only a single Metropolis proposal.

The constructing of formal optimization criteria for more elaborate implementations of the 
algorithm, including windowed samplers \citep{Neal:1994}, such as the No-U-Turn Sampler 
\citep{HoffmanEtAl:2014}, that subsample from each approximate trajectory, can also be placed 
into this geometric framework and utilize the theorems proved in this paper.  Similarly, the
understanding of canonical expectations we have built in this paper is applicable to 
Rao-Blackwellization schemes that keep all points along each approximate trajectory, using 
weights to correct for the error in the symplectic integrator. 

\section{Acknowledgements}

We warmly thank Chris Wendl for illuminating discussions on various aspects of symplectic 
geometry used in the proofs and Elena Akhmatskaya for helpful comments.  
Michael Betancourt is supported under EPSRC grant EP/J016934/1, 
Simon Byrne is a EPSRC Postdoctoral Research Fellow under grant EP/K005723/1, 
and Mark Girolami is an EPSRC Established Career Research Fellow under grant EP/J016934/1.

\setcounter{section}{0}
\renewcommand{\thesection}{\Alph{section}}

\section{Proofs} \label{apx:main}

Here we present proofs for the Lemmas and Theorems appearing in Sections 
\ref{sec:accuracy_of_canonical_expectations} and
\ref{sec:accuracy_of_canonical_expectations_of_hamiltonian_error}, as well 
as a rigorous calculation of the average Hamiltonian error for the example in
Section \ref{sec:accuracy_of_canonical_expectations_of_hamiltonian_error}.

\subsection{The Accuracy of Canonical Expectations}
\label{apx:accuracy}

\newtheorem*{thm:micro_canonical}{Theorem \ref{thm:micro_canonical}}
\begin{thm:micro_canonical}

Let $\left( M, \omega, H\right)$ be a Hamiltonian system and consider
a $k$-th order symmetric symplectic integrator with the corresponding modified Hamiltonian
$\widetilde{H} = H + \epsilon^{k} \, G + \mathcal{O} ( \epsilon^{k + 2} )$.  
If the integrator is topologically stable and the asymptotic error is negligible, 
then the difference in the microcanonical and modified microcanonical expectations for 
any smooth function, $f : M \rightarrow \mathbb{R}$, is given by
\begin{align*}
\left< f \right>_{\sLS}- \left< f \right>_{\sMLS} 
=&
+ \epsilon^{k}
\left< 
f \, \mathrm{d} G \! \left( \vec{v} \right) + G \, \mathrm{d} f \! \left( \vec{v} \right) 
+ f \, G \left( \frac{\partial v^{i} }{ \partial q^{i} } + \frac{\partial v_{i} }{ \partial p_{i} } \right) 
\right>_{ \sMLS } 
\\
& - \epsilon^{k}
\Big< f \Big>_{\sMLS} 
\left< 
\mathrm{d} G \! \left( \vec{v} \right)
+ G \left( \frac{\partial v^{i} }{ \partial q^{i} } + \frac{\partial v_{i} }{ \partial p_{i} } \right) 
\right>_{ \sMLS }
+ \mathcal{O} ( \epsilon^{k + 2} ),
\end{align*}
where $\vec{v}$ is any transverse vector field satisfying $\mathrm{d} H \! \left( \vec{v} \right) = 1$.

\end{thm:micro_canonical}

\begin{figure}
\setlength{\unitlength}{0.1in} 
\centering
\begin{picture}(45, 30)
\put(2.55, 0) { \includegraphics[width=4in]{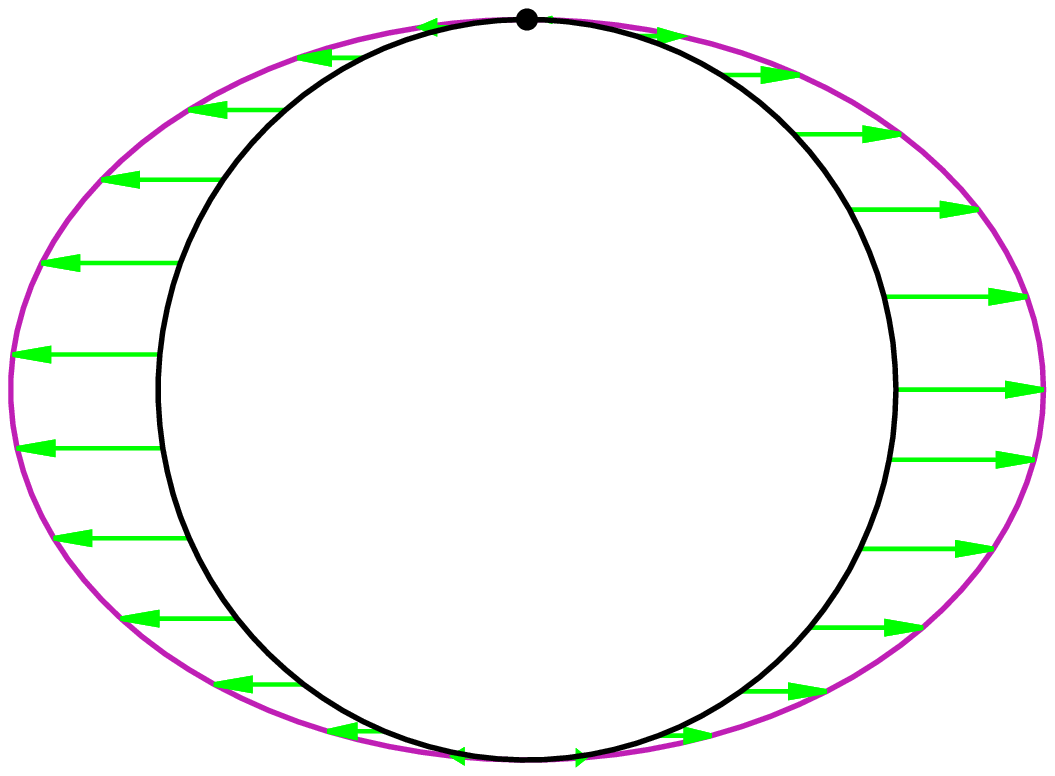} }
\put(14, 16) { \makebox(0, 0) { $\LS$ } }
\put(2, 16) { \makebox(0, 0) { $\MLS$ } }
\put(35.5, 11.5) { \makebox(0, 0){ $G \, \vec{v}$ } }
\end{picture} 
\caption{Geometrically, a level set \LS can be warped into a level set \MLS by dragging it 
along the flow generated by $G \, \vec{v}$, admitting a comparison of expectations along 
the two surfaces.}
\label{fig:drag} 
\end{figure}

\begin{proof}

Provided that the integrator is stable and at the topologies of the level sets are the same,
we can compare the two expectations by perturbing the true level set into the modified 
one by dragging it along $G \, \vec{v}$ (Figure \ref{fig:drag}).  In particular, dragging the 
integrand gives
\begin{align*}
\left( f \, \vec{v} \cOmega \right)'
&= \int_{0}^{\epsilon^{k}} \mathrm{d}t \, \mathcal{L}_{ G \, \vec{v} } \left( f \, \vec{v} \cOmega \right)
\\
&= f \, \vec{v} \cOmega 
+ \epsilon^{k} \, \mathcal{L}_{ G \, \vec{v}} \left( f \, \vec{v} \cOmega \right) 
+ \mathcal{O} ( \epsilon^{k + 2} ).
\end{align*}
  
The Lie derivative evaluates to
\begin{align*}
\mathcal{L}_{ G \, \vec{v}} \left( f \, \vec{v} \cOmega \right) 
&= \mathcal{L}_{G \, \vec{v}} \left( f \right) \vec{v} \cOmega
+ f \mathcal{L}_{G \, \vec{v}} \left( \vec{v} \cOmega \right) 
\\
&= \mathcal{L}_{G \, \vec{v}} \left( f \right) \vec{v} \cOmega
+ f \big[ G \, \vec{v}, \vec{v} \big] \cOmega
+ f \, \vec{v} \, \lrcorner \, \mathcal{L}_{G \, \vec{v}} \, \Omega,
\end{align*}
where
\begin{align*}
\mathcal{L}_{G \, \vec{v}} \left( f \right)
&= G \, \mathrm{d} f \! \left( \vec{v} \right),
\\
\\
\big[ G \, \vec{v}, \vec{v} \big] 
&= G \big[ \vec{v}, \vec{v} \big] - \mathrm{d} G \! \left( \vec{v} \right) \vec{v} \\
&= - \mathrm{d} G \! \left( \vec{v} \right) \vec{v},
\intertext{and}
\mathcal{L}_{G \, \vec{v}} \, \Omega
&= G \, \mathcal{L}_{\vec{v} } \, \Omega + \mathrm{d} G \wedge \left( \vec{v} \cOmega \right) \\
&= G \, \vec{v} \, \lrcorner \, \mathrm{d} \Omega + G \, \mathrm{d} \! \left( \vec{v} \cOmega \right) 
+ \mathrm{d} G \wedge \left( \vec{v} \cOmega \right) \\
&= G \, \mathrm{d} \! \left( \vec{v} \cOmega \right) 
+ \mathrm{d} G \wedge \left( \vec{v} \cOmega \right) \\
&= G \, \mathrm{d} \! \left( \vec{v} \cOmega \right) 
+ \mathrm{d} G \! \left( \vec{v} \right) \cOmega
- \vec{v} \, \lrcorner \, \left( \mathrm{d} G \wedge \Omega \right) \\
&= G \, \mathrm{d} \! \left( \vec{v} \cOmega \right) 
+ \mathrm{d} G \! \left( \vec{v} \right) \cOmega.
\end{align*}
Hence
\begin{align*}
\mathcal{L}_{ G \, \vec{v} } \left( f \vec{v} \cOmega \right)
&= G \, \mathrm{d} f \! \left( \vec{v} \right) \vec{v} \cOmega
 - f \, \mathrm{d}G \! \left( \vec{v} \right) \vec{v} \cOmega
+ f \, G \, \vec{v} \, \lrcorner \, \mathrm{d} \left( \vec{v} \cOmega \right) 
+ f \, \mathrm{d} G\! \left( \vec{v} \right) \vec{v} \cOmega \\
&= G \, \mathrm{d} f \! \left( \vec{v} \right) \vec{v} \cOmega
+ f \, G \, \vec{v} \, \lrcorner \, \mathrm{d} \left( \vec{v} \cOmega \right),
\end{align*}
or, in canonical coordinates,
\begin{equation*}
\mathcal{L}_{ G \, \vec{v} } \left( f \vec{v} \cOmega \right)
= G \, \mathrm{d} f \! \left( \vec{v} \right) \vec{v} \cOmega
+ f \, G \left( \frac{\partial v^{i} }{ \partial q^{i} } + \frac{\partial v_{i} }{ \partial p_{i} } \right) \vec{v} \cOmega.
\end{equation*}
Together these give
\begin{equation*}
\left( f \, \vec{v} \cOmega \right)'
= \left( 
f 
+ \epsilon^{k} \, G \, \mathrm{d} f \! \left( \vec{v} \right) 
+ \epsilon^{k} \, f \, G \left( \frac{\partial v^{i} }{ \partial q^{i} } + \frac{\partial v_{i} }{ \partial p_{i} } \right) \right)
\vec{v} \cOmega 
+ \mathcal{O} ( \epsilon^{k + 2} ),
\end{equation*}

Before we can pull these terms onto $\MLS$ we have to relate them to the proper volume form, 
$\vec{u} \cOmega$.  Because
\begin{align*}
\vec{v} \cOmega 
&= \mathrm{d} \widetilde{H} \! \left( \vec{v} \right) \vec{u} \cOmega \\
&= \left( \mathrm{d} H \! \left( \vec{v} \right) + \epsilon^{k} \,  \mathrm{d} G \! \left( \vec{v} \right) \right) 
\vec{u} \cOmega \\
&= \left( 1 + \epsilon^{k} \, \mathrm{d} G \! \left( \vec{v} \right) \right) \vec{u} \cOmega,
\end{align*}
we have, to leading-order in the step size,
\begin{equation*}
\left( f \, \vec{v} \cOmega \right)'
= f \vec{u} \cOmega
+ \epsilon^{k} \left( f\, \mathrm{d} G \! \left( \vec{v} \right) + G \, \mathrm{d} f \! \left( \vec{v} \right) 
+ f\, G \left( \frac{\partial v^{i} }{ \partial q^{i} } + \frac{\partial v_{i} }{ \partial p_{i} } \right) \right)
\vec{u} \cOmega 
+ \mathcal{O} ( \epsilon^{k + 2} ),
\end{equation*}
Pulling back onto the modified level set gives
\begin{align*}
\int_{ \sLS } \iota^{*}_{E} \left( f \, \vec{v} \cOmega \right) 
=& 
\int_{ \sMLS } 
\tilde{\iota}^{*}_{E} \left( f \, \vec{v} \cOmega \right)' 
\\
=& 
\int_{ \sMLS }
\tilde{\iota}^{*}_{E} \! \left( f\, \vec{u} \cOmega \right) 
\\
&
+ \epsilon^{k} \int_{ \sMLS }
\tilde{\iota}^{*}_{E} \!  \left( \left( f\, \mathrm{d} G \! \left( \vec{v} \right) + G \, \mathrm{d} f \! \left( \vec{v} \right) 
+ f\, G \left( \frac{\partial v^{i} }{ \partial q^{i} } + \frac{\partial v_{i} }{ \partial p_{i} } \right) \right)
 \vec{u} \cOmega \right) 
+ \mathcal{O} ( \epsilon^{k + 2} ),
\end{align*}

Consequently the microcanonical expectation becomes
\begin{align*}
\left< f \right>_{ \sLS }
=& 
\frac{ \int_{\sLS} \iota^{*}_{E} \! \left( f \, \vec{v} \cOmega \right) }
{ \int_{\sLS} \iota^{*}_{E} \! \left( \vec{v} \cOmega \right) } \\
=&
\frac
{ \int_{ \sMLS } \tilde{\iota}^{*}_{E} \! \left( f \, \vec{u} \cOmega \right)
+ \epsilon^{k} \int_{ \sMLS }
\tilde{\iota}^{*}_{E} \! \left( \left( f \, \mathrm{d} G \! \left( \vec{v} \right) + G \, \mathrm{d} f \! \left( \vec{v} \right) 
+ f\, G \left( \frac{\partial v^{i} }{ \partial q^{i} } + \frac{\partial v_{i} }{ \partial p_{i} } \right) \right)
 \vec{u} \cOmega \right) }
{ \int_{ \sMLS } \tilde{\iota}^{*}_{E} \! \left( \vec{u} \cOmega \right)
+ \epsilon^{k} \int_{ \sMLS }
\tilde{\iota}^{*}_{E} \!  \left( \left( \mathrm{d} G \! \left( \vec{v} \right)
+ G \left( \frac{\partial v^{i} }{ \partial q^{i} } + \frac{\partial v_{i} }{ \partial p_{i} } \right) \right)
\vec{u} \cOmega \right) } 
\\
&+ \mathcal{O} ( \epsilon^{2k} )
\\
=&
\frac
{ \int_{ \sMLS } \tilde{\iota}^{*}_{E} \! \left( f \, \vec{u} \cOmega \right) }
{ \int_{ \sMLS } \tilde{\iota}^{*}_{E} \! \left( \vec{u} \cOmega \right) } \\
&+ \epsilon^{k}
\frac
{ \int_{ \sMLS }
\tilde{\iota}^{*}_{E} \!  \left( \left( f\, \mathrm{d} G \! \left( \vec{v} \right) + G \, \mathrm{d} f \! \left( \vec{v} \right) 
+ f\, G \left( \frac{\partial v^{i} }{ \partial q^{i} } + \frac{\partial v_{i} }{ \partial p_{i} } \right) \right)
 \vec{u} \cOmega \right) }
{ \int_{ \sMLS } \tilde{\iota}^{*}_{E} \! \left( \vec{u} \cOmega \right) } \\
& - \epsilon^{k}
\frac 
{ \int_{ \sMLS } \tilde{\iota}^{*}_{E} \! \left( f \, \vec{u} \cOmega \right) }
{ \int_{ \sMLS } \tilde{\iota}^{*}_{E} \! \left( \vec{u} \cOmega \right) }
\frac
{ \int_{ \sMLS }
\tilde{\iota}^{*}_{E} \!  \left( \left( \mathrm{d} G \! \left( \vec{v} \right)
+ G \left( \frac{\partial v^{i} }{ \partial q^{i} } + \frac{\partial v_{i} }{ \partial p_{i} } \right) \right)
 \vec{u} \cOmega \right) }
{ \int_{ \sMLS } \tilde{\iota}^{*}_{E} \! \left( \vec{u} \cOmega \right) }
+ \mathcal{O} ( \epsilon^{k + 2} ) 
\\
=&
\left< f \right>_{ \sMLS }
+ \epsilon^{k}
\left< 
f \, \mathrm{d} G \! \left( \vec{v} \right) + G \, \mathrm{d} f \! \left( \vec{v} \right) 
+ f\, G \left( \frac{\partial v^{i} }{ \partial q^{i} } + \frac{\partial v_{i} }{ \partial p_{i} } \right) 
\right>_{ \sMLS } 
\\
& - \epsilon^{k}
\left< f \right>_{ \sMLS}
\left< 
\mathrm{d} G \! \left( \vec{v} \right)
+ G \left( \frac{\partial v^{i} }{ \partial q^{i} } + \frac{\partial v_{i} }{ \partial p_{i} } \right) 
\right>_{ \sMLS }
+ \mathcal{O} ( \epsilon^{k + 2} ),
\end{align*}
or
\begin{align*}
\left< f \right>_{ \sLS } - \left< f \right>_{ \sMLS } 
=& 
+ \epsilon^{k}
\left< 
f \, \mathrm{d} G \! \left( \vec{v} \right) + G \, \mathrm{d} f \! \left( \vec{v} \right) 
+ f\, G \left( \frac{\partial v^{i} }{ \partial q^{i} } + \frac{\partial v_{i} }{ \partial p_{i} } \right) 
\right>_{ \sMLS } \\
& - \epsilon^{k}
\left< f \right>_{ \sMLS } 
\left< 
\mathrm{d} G \! \left( \vec{v} \right)
+ G \left( \frac{\partial v^{i} }{ \partial q^{i} } + \frac{\partial v_{i} }{ \partial p_{i} } \right) 
\right>_{ \sMLS }
+ \mathcal{O} ( \epsilon^{k + 2} ),
\end{align*}
as desired.

\end{proof}

\begin{corollary}
\label{cor:density_of_states}
Let $\left( M, \omega, H\right)$ be a Hamiltonian system and consider
a $k$-th order symmetric symplectic integrator with the corresponding modified Hamiltonian
$\widetilde{H} = H + \epsilon^{k} \, G + \mathcal{O} ( \epsilon^{k + 2} )$. If the 
integrator is topologically stable and the asymptotic error is negligible then
\begin{align*}
\frac{1}{ \int_{M} e^{-\beta H} \Omega }
&=
\frac{1}{ \int_{M} e^{-\beta \widetilde{H} } \Omega }
\left(1 -
\epsilon^{k} \, \mathbb{E}_{\varpi_{\widetilde{H}}} \! \left[
\mathrm{d} G \! \left( \vec{v} \right)
+ G \left( \frac{\partial v^{i} }{ \partial q^{i} } + \frac{\partial v_{i} }{ \partial p_{i} } \right) 
\right] \right)
+ \mathcal{O} ( \epsilon^{k + 2} )
\end{align*}
where $\vec{v}$ is any transverse vector field satisfying $\mathrm{d} H \! \left( \vec{v} \right) = 1$.

\end{corollary}

\begin{proof}

Taking $f = 1$ in the derivation Theorem~\ref{thm:micro_canonical} we have
\begin{align*}
d \! \left( E \right) 
&= 
\int_{ \sLS } \iota^{*}_{E} \! \left( \vec{v} \cOmega \right)
\\
&= 
\int_{ \sMLS }
\tilde{\iota}^{*}_{E} \! \left( \vec{u} \cOmega \right) 
\\
& \quad
+ \epsilon^{k} \int_{ \sMLS }
\tilde{\iota}^{*}_{E} \!  \left( \left( \mathrm{d} G \! \left( \vec{v} \right)
+ G \left( \frac{\partial v^{i} }{ \partial q^{i} } + \frac{\partial v_{i} }{ \partial p_{i} } \right) \right)
 \vec{u} \cOmega \right) 
+ \mathcal{O} ( \epsilon^{k + 2} )
\\
&\equiv \tilde{d} \! \left( E \right) + \epsilon^{k} \, d' \! \left( E \right) + \mathcal{O} ( \epsilon^{k + 2} ).
\end{align*}

Then
\begin{align*}
\frac{1}{ \int_{M} e^{-\beta H} \Omega }
&=
\frac{1}{ \int \mathrm{d} E \, d \! \left( E \right) e^{-\beta E} }
\\
&=
\frac{1}{ \int \mathrm{d} E \, \left( 
\tilde{d} \! \left( E \right) + \epsilon^{k} \, d' \! \left( E \right) + \mathcal{O} ( \epsilon^{k + 2} ) 
\right) e^{-\beta E} }
\\
&=
\frac{1}{ 
\int \mathrm{d} E \, \tilde{d} \! \left( E \right) e^{-\beta E}
+ \epsilon^{k} \, \int \mathrm{d} E \, e^{-\beta E}
d' \! \left( E \right)  + \mathcal{O} ( \epsilon^{k + 2} )
}
\\
&=
\frac{1}{ \int \mathrm{d} E \, \tilde{d} \! \left( E \right) e^{-\beta E} }
\left( 1
- \epsilon^{k} \, 
\frac{\int \mathrm{d} E \, e^{-\beta E} d' \! \left( E \right) }
{ \int \mathrm{d} E \, \tilde{d} \! \left( E \right) e^{-\beta E} }
 \right)
+ \mathcal{O} ( \epsilon^{k + 2} )
\\
&=
\frac{1}{ \int \mathrm{d} E \, \tilde{d} \! \left( E \right) e^{-\beta E} }
\left( 1 - \epsilon^{k} \, \mathbb{E}_{\varpi_{\widetilde{H}}} \! \left[ d' \! \left( E \right) \right] \right)
+ \mathcal{O} ( \epsilon^{k + 2} )
\\
&=
\frac{1}{ \int_{M} e^{-\beta \widetilde{H}} \Omega }
\left( 1 - \epsilon^{k} \, \mathbb{E}_{\varpi_{\widetilde{H}}} \! \left[ d' \! \left( E \right) \right] \right)
+ \mathcal{O} ( \epsilon^{k + 2} )
\\
\intertext{or}
\frac{1}{ \int_{M} e^{-\beta H} \Omega }
&=
\frac{1}{ \int_{M} e^{-\beta \widetilde{H}} \Omega }
\left( 1 - \epsilon^{k} \, \mathbb{E}_{\varpi_{\widetilde{H}}} \! \left[ 
 \mathrm{d} G \! \left( \vec{v} \right)
+ G \left( \frac{\partial v^{i} }{ \partial q^{i} } + \frac{\partial v_{i} }{ \partial p_{i} } \right) 
 \right] \right)
+ \mathcal{O} ( \epsilon^{k + 2} ),
\end{align*}
as desired.

\end{proof}

\newtheorem*{thm:canonical}{Theorem \ref{thm:canonical}}
\begin{thm:canonical}

Let $\left( M, \omega, H\right)$ be a Hamiltonian system and consider
a $k$-th order symmetric symplectic integrator with the corresponding modified Hamiltonian
$\widetilde{H} = H + \epsilon^{k} \, G + \mathcal{O} ( \epsilon^{k + 2} )$.  
If the integrator is topologically stable and the asymptotic error is negligible, 
then the difference in the canonical and modified canonical expectations for any smooth
function, $f : M \rightarrow \mathbb{R}$, is given by
\begin{align*}
\mathbb{E}_{\varpi_{H}} \! \left[ f \right] - \mathbb{E}_{\varpi_{\widetilde{H}}} \! \left[ f \right]
=&
+ \epsilon^{k} \,
\mathbb{E}_{\varpi_{H}} \! \left[
f \, \mathrm{d} G \! \left( \vec{v} \right) + G \, \mathrm{d} f \! \left( \vec{v} \right) 
+ f \, G \left( \frac{\partial v^{i} }{ \partial q^{i} } + \frac{\partial v_{i} }{ \partial p_{i} } \right) 
\right] 
\\
& - \epsilon^{k} \,
\mathbb{E}_{\varpi_{H}} \! \left[ f \right] \,
\mathbb{E}_{\varpi_{H}} \! \left[
\mathrm{d} G \! \left( \vec{v} \right)
+ G \left( \frac{\partial v^{i} }{ \partial q^{i} } + \frac{\partial v_{i} }{ \partial p_{i} } \right) \right] 
+ \mathcal{O} ( \epsilon^{k + 2} ),
\end{align*}
where $\vec{v}$ is any transverse vector field satisfying $\mathrm{d} H \! \left( \vec{v} \right) = 1$.

\end{thm:canonical}

\begin{proof}

Applying Theorem~\ref{thm:micro_canonical},
\begin{align*}
\int_{ \sLS } \iota^{*}_{E} \left( f \, \vec{v} \cOmega \right) 
&= 
\int_{ \sMLS } \tilde{\iota}^{*}_{E} \! \left( f\, \vec{u} \cOmega \right) 
\\
& \quad
+ \epsilon^{k} \int_{ \sMLS }
\tilde{\iota}^{*}_{E} \!  \left( \left( f\, \mathrm{d} G \! \left( \vec{v} \right) + G \, \mathrm{d} f \! \left( \vec{v} \right) 
+ f\, G \left( \frac{\partial v^{i} }{ \partial q^{i} } + \frac{\partial v_{i} }{ \partial p_{i} } \right) \right)
 \vec{u} \cOmega \right) 
\\
& \quad
+ \mathcal{O} ( \epsilon^{k + 2} )
\\
&\equiv \mathcal{I} + \epsilon^{k} \, \mathcal{I}' + \mathcal{O} ( \epsilon^{k + 2} ),
\end{align*}
and
\begin{align*}
d \! \left( E \right) 
&= 
\int_{ \sLS } \iota^{*}_{E} \! \left( \vec{v} \cOmega \right)
\\
&= 
\int_{ \sMLS }
\tilde{\iota}^{*}_{E} \! \left( \vec{u} \cOmega \right) 
\\
& \quad
+ \epsilon^{k} \int_{ \sMLS }
\tilde{\iota}^{*}_{E} \!  \left( \left( \mathrm{d} G \! \left( \vec{v} \right)
+ G \left( \frac{\partial v^{i} }{ \partial q^{i} } + \frac{\partial v_{i} }{ \partial p_{i} } \right) \right)
 \vec{u} \cOmega \right) 
+ \mathcal{O} ( \epsilon^{k + 2} )
\\
&\equiv \tilde{d} \! \left( E \right) + \epsilon^{k} \, d' \! \left( E \right) + \mathcal{O} ( \epsilon^{k + 2} ).
\end{align*}

Substituting these into the definition of the canonical expectation gives
\begin{align*}
\mathbb{E}_{\varpi_{H}} \! \left[ f \right]
&=
\frac{\int \mathrm{d} E \, e^{-\beta E} 
\int_{ \sLS } \iota^{*}_{E} \! \left( f \, \vec{v} \cOmega \right)}
{\int \mathrm{d} E \, d \! \left( E \right) e^{-\beta E}}
\\
&=
\frac{
\int \mathrm{d} E \, e^{-\beta E} 
\left( \mathcal{I} + \epsilon^{k} \, \mathcal{I}' + \mathcal{O} ( \epsilon^{k + 2} ) \right)
}{
\int \mathrm{d} E \, e^{-\beta E}
\left( \tilde{d} \! \left( E \right) + \epsilon^{k} \, d' \! \left( E \right) + \mathcal{O} ( \epsilon^{k + 2} ) \right)
}
\\
&=
\frac{ \int \mathrm{d} E \, e^{-\beta E} \mathcal{I} }
{ \int \mathrm{d} E \, e^{-\beta E} \tilde{d} \! \left( E \right) }
+ \delta
\frac{ \int \mathrm{d} E \, e^{-\beta E} \mathcal{I}' }
{ \int \mathrm{d} E \, e^{-\beta E} \tilde{d} \! \left( E \right) }
- \delta
\frac{ \int \mathrm{d} E \, e^{-\beta E} \mathcal{I} }
{ \int \mathrm{d} E \, e^{-\beta E} \tilde{d} \! \left( E \right) }
\frac{ \int \mathrm{d} E \, e^{-\beta E} d' \! \left( E \right) }
{ \int \mathrm{d} E \, e^{-\beta E} \tilde{d} \! \left( E \right) }
+ \mathcal{O} ( \epsilon^{k + 2} )
\\
&= +
\frac{ \int \mathrm{d} E \, e^{-\beta E} \int_{ \sMLS } \tilde{\iota}^{*}_{E} \! \left( f\, \vec{u} \cOmega \right) }
{ \int \mathrm{d} E \, \tilde{d} \! \left( E \right) e^{-\beta E} }
\\
& \quad + \epsilon^{k}
\frac{ \int \mathrm{d} E \, e^{-\beta E} \epsilon^{k} \int_{ \sMLS }
\tilde{\iota}^{*}_{E} \!  \left( \left( f\, \mathrm{d} G \! \left( \vec{v} \right) + G \, \mathrm{d} f \! \left( \vec{v} \right) 
+ f\, G \left( \frac{\partial v^{i} }{ \partial q^{i} } + \frac{\partial v_{i} }{ \partial p_{i} } \right) \right)
 \vec{u} \cOmega \right) }
{ \int \mathrm{d} E \, \tilde{d} \! \left( E \right) e^{-\beta E} }
\\
& \quad - \epsilon^{k}
\frac{ \int \mathrm{d} E \, e^{-\beta E} \int_{ \sMLS } \tilde{\iota}^{*}_{E} \! \left( f\, \vec{u} \cOmega \right) }
{ \int \mathrm{d} E \, \tilde{d} \! \left( E \right) e^{-\beta E} }
\frac{ \int \mathrm{d} E \, e^{-\beta E} \int_{ \sMLS }
\tilde{\iota}^{*}_{E} \!  \left( \left( \mathrm{d} G \! \left( \vec{v} \right)
+ G \left( \frac{\partial v^{i} }{ \partial q^{i} } + \frac{\partial v_{i} }{ \partial p_{i} } \right) \right)
 \vec{u} \cOmega \right)  }
{ \int \mathrm{d} E \, \tilde{d} \! \left( E \right) e^{-\beta E} }
\\
& \quad + \mathcal{O} ( \epsilon^{k + 2} )
\\
&= +
\mathbb{E}_{\varpi_{\widetilde{H}}} \! \Big[ f \Big]
+ \epsilon^{k} \,
\mathbb{E}_{\varpi_{\widetilde{H}}} \! \left[ f\, \mathrm{d} G \! \left( \vec{v} \right) + G \, \mathrm{d} f \! \left( \vec{v} \right) 
+ f\, G \left( \frac{\partial v^{i} }{ \partial q^{i} } + \frac{\partial v_{i} }{ \partial p_{i} } \right) \right]
\\
& \quad - \epsilon^{k} \,
\mathbb{E}_{\varpi_{\widetilde{H}}} \! \Big[ f \Big]
\mathbb{E}_{\varpi_{\widetilde{H}}} \! \left[ \mathrm{d} G \! \left( \vec{v} \right)
+ G \left( \frac{\partial v^{i} }{ \partial q^{i} } + \frac{\partial v_{i} }{ \partial p_{i} } \right) \right]
+ \mathcal{O} ( \epsilon^{k + 2} ),
\end{align*}
or
\begin{align*}
\mathbb{E}_{\varpi_{H}} \! \left[ f \right] - \mathbb{E}_{\varpi_{\widetilde{H}}} \! \left[ f \right]&=
+ \epsilon^{k} \,
\mathbb{E}_{\varpi_{\widetilde{H}}} \! \left[ f\, \mathrm{d} G \! \left( \vec{v} \right) + G \, \mathrm{d} f \! \left( \vec{v} \right) 
+ f\, G \left( \frac{\partial v^{i} }{ \partial q^{i} } + \frac{\partial v_{i} }{ \partial p_{i} } \right) \right]
\\
& \quad - \epsilon^{k} \,
\mathbb{E}_{\varpi_{\widetilde{H}}} \! \Big[ f \Big]
\mathbb{E}_{\varpi_{\widetilde{H}}} \! \left[ \mathrm{d} G \! \left( \vec{v} \right)
+ G \left( \frac{\partial v^{i} }{ \partial q^{i} } + \frac{\partial v_{i} }{ \partial p_{i} } \right) \right]
+ \mathcal{O} ( \epsilon^{k + 2} ),
\end{align*}
as desired.

\end{proof}

\subsection{Approximating Canonical Expectations of the Hamiltonian Error}
\label{apx:approx}

In order to compute moments of the Hamiltonian error we first have to be able to manipulate 
the Metropolis proposal operator, $\Phi^{\widetilde{H}}_{\epsilon, \tau}$.  Fortunately, the 
operator is a diffeomorphism, $\Phi^{\widetilde{H}}_{\epsilon, \tau} : M \rightarrow M$ which 
admits a variety of convenient manipulations.

\begin{lemma}
\label{lem:correlation_by_parts}
``Correlation by Parts''.
Let $\left( M, \omega, H\right)$ be a Hamiltonian system and consider
a $k$-th order symmetric symplectic integrator with the corresponding modified Hamiltonian
$\widetilde{H} = H + \epsilon^{k} \, G + \mathcal{O} ( \epsilon^{k + 2} )$
and Metropolis proposal $\Phi^{\widetilde{H}}_{\epsilon, \tau}$. 
If the integrator is topologically stable and the asymptotic error is negligible, 
then the correlation of two smooth functions, $f : M \rightarrow \mathbb{R}$
and $h : M \rightarrow \mathbb{R}$, with respect to a unit canonical distribution satisfies
\begin{align*}
\mathbb{E}_{\varpi_{H}} \! \left[ f \left( h \circ \Phi^{\widetilde{H}}_{\epsilon, \tau} \right) \right] 
&= 
\mathbb{E}_{\varpi_{\widetilde{H}}} \! \left[ \left( f \circ \Phi^{\widetilde{H}}_{\epsilon, \tau} \right) h \right]
\\
& \quad + \epsilon^{k} \Big( 
- \mathbb{E}_{\varpi_{H}} \! \left[ f \left( h \circ \Phi^{\widetilde{H}}_{\epsilon, \tau} \right) \right]
\mathbb{E}_{\varpi_{\widetilde{H}}} \! \left[ \mathrm{d} G \! \left( \vec{v} \right)
+ G \left( \frac{\partial v^{i} }{ \partial q^{i} } + \frac{\partial v_{i} }{ \partial p_{i} } \right) \right]
\\
& \quad\quad\quad\; + \mathbb{E}_{\varpi_{\widetilde{H}}} \! \left[ 
\left( f \circ \Phi^{\widetilde{H}}_{\epsilon, \tau} \right) h \, \left( G \circ \Phi^{\widetilde{H}}_{\epsilon, \tau} \right) 
\right] \Big)
\\
& \quad + \mathcal{O} ( \epsilon^{k + 2} ),
\end{align*}
where $\vec{v}$ is any transverse vector field satisfying $\mathrm{d} H \! \left( \vec{v} \right) = 1$.

\end{lemma}

\begin{proof}

By definition
\begin{equation*}
\mathbb{E}_{\varpi_{H}} \! \left[ f \left( h \circ \Phi^{\widetilde{H}}_{\epsilon, \tau} \right) \right] =
\frac{
\int_{M} f \left( h \circ \Phi^{\widetilde{H}}_{\epsilon, \tau} \right) e^{-H} \Omega
}{
\int_{M} e^{- H} \Omega
}
\end{equation*}
Because $\Phi^{\widetilde{H}}_{\epsilon, \tau}$ is a diffeomorphism we can push-forward
integrals; in particular the numerator becomes
\begin{align*}
\int_{M} f \left( h \circ \Phi^{\widetilde{H}}_{\epsilon, \tau} \right) e^{- H} \Omega
&=
\int_{M} \left( \Phi^{\widetilde{H}}_{\epsilon, \tau} \right)^{*} 
\left( f \left( h \circ \Phi^{\widetilde{H}}_{\epsilon, \tau} \right) e^{-H} \Omega \right)
\\
&=
\int_{M} \left( f \circ \Phi^{\widetilde{H}}_{\epsilon, \tau} \right) 
\left( h \circ \Phi^{\widetilde{H}}_{\epsilon, \tau} \circ \Phi^{\widetilde{H}}_{\epsilon, \tau} \right) 
e^{-H \circ \Phi^{\widetilde{H}}_{\epsilon, \tau} } \Omega
\\
&=
\int_{M} \left( f \circ \Phi^{\widetilde{H}}_{\epsilon, \tau} \right) h \, 
e^{-H \circ \Phi^{\widetilde{H}}_{\epsilon, \tau} } \Omega,
\end{align*}
where we have used the fact that $\Phi^{\widetilde{H}}_{\epsilon, \tau}$ is idempotent.  Substituting 
the definition of the modified Hamiltonian into the exponent and expanding then gives
\begin{align*}
\int_{M} f \left( h \circ \Phi^{\widetilde{H}}_{\epsilon, \tau} \right) e^{- H} \Omega
&=
\int_{M} \left( f \circ \Phi^{\widetilde{H}}_{\epsilon, \tau} \right) h \, 
e^{- H \circ \Phi^{\widetilde{H}}_{\epsilon, \tau} } \Omega
\\
&=
\int_{M} \left( f \circ \Phi^{\widetilde{H}}_{\epsilon, \tau} \right) 
h \, 
e^{- \left( \widetilde{H} - \epsilon^{k} \, G \right) \circ \Phi^{\widetilde{H}}_{\epsilon, \tau} } \Omega
\\
&=
\int_{M} \left( f \circ \Phi^{\widetilde{H}}_{\epsilon, \tau} \right) h \, 
e^{- \widetilde{H} \circ \Phi^{\widetilde{H}}_{\epsilon, \tau} }
e^{ \epsilon^{k} \, G \circ \Phi^{\widetilde{H}}_{\epsilon, \tau} } \Omega
\\
&=
\int_{M} \left( f \circ \Phi^{\widetilde{H}}_{\epsilon, \tau} \right) h \, 
e^{- \widetilde{H} \circ \Phi^{\widetilde{H}}_{\epsilon, \tau} }
\left( 1 + \epsilon^{k} \, G \circ \Phi^{\widetilde{H}}_{\epsilon, \tau} \right) \Omega
+ \mathcal{O} ( \epsilon^{k + 2} )
\\
&=
\int_{M} \left( f \circ \Phi^{\widetilde{H}}_{\epsilon, \tau} \right) 
h \, e^{- \widetilde{H} \circ \Phi^{\widetilde{H}}_{\epsilon, \tau} } \Omega
\\
& \quad +
\epsilon^{k} \int_{M} \left( f \circ \Phi^{\widetilde{H}}_{\epsilon, \tau} \right) h \, \left( G \circ \Phi^{\widetilde{H}}_{\epsilon, \tau} \right)
e^{- \widetilde{H} \circ \Phi^{\widetilde{H}}_{\epsilon, \tau} } \Omega
+ \mathcal{O} ( \epsilon^{k + 2} ).
\end{align*}
But the modified Hamiltonian is invariant to the modified flow,
$\widetilde{H} \circ \Phi^{\widetilde{H}}_{\epsilon, \tau} = \widetilde{H}$ and the integrals on the RHS
become
\begin{align*}
\int_{M} f \left( h \circ \Phi^{\widetilde{H}}_{\epsilon, \tau} \right) e^{-H} \Omega
&=
\int_{M} \left( f \circ \Phi^{\widetilde{H}}_{\epsilon, \tau} \right) 
h \, e^{- \widetilde{H} } \Omega
\\
& \quad +
\epsilon^{k} \int_{M} \left( f \circ \Phi^{\widetilde{H}}_{\epsilon, \tau} \right) h \, \left( G \circ \Phi^{\widetilde{H}}_{\epsilon, \tau} \right)
e^{- \widetilde{H} } \Omega
+ \mathcal{O} ( \epsilon^{k + 2} ).
\end{align*}

Dividing by the modified normalization,
\begin{align*}
\frac{
\int_{M} f \left( h \circ \Phi^{\widetilde{H}}_{\epsilon, \tau} \right) e^{-H} \Omega
}{
\int_{M} e^{- \widetilde{H} } \Omega
}
&=
\frac{
\int_{M} \left( f \circ \Phi^{\widetilde{H}}_{\epsilon, \tau} \right) 
h \, e^{- \widetilde{H} } \Omega
}{
\int_{M} e^{- \widetilde{H} } \Omega
}
+ \epsilon^{k} 
\frac{
\int_{M} 
\left( f \circ \Phi^{\widetilde{H}}_{\epsilon, \tau} \right) h \, 
\left( G \circ \Phi^{\widetilde{H}}_{\epsilon, \tau} \right)
e^{- \widetilde{H} } \Omega
}{
\int_{M} e^{- \widetilde{H} } \Omega
}
+ \mathcal{O} ( \epsilon^{k + 2} ) 
\\
&=
\mathbb{E}_{\varpi_{\widetilde{H}}} \! \left[ \left( f \circ \Phi^{\widetilde{H}}_{\epsilon, \tau} \right) h \right]
+ \epsilon^{k} \,
\mathbb{E}_{\varpi_{\widetilde{H}}} \! \left[ 
\left( f \circ \Phi^{\widetilde{H}}_{\epsilon, \tau} \right) h \, \left( G \circ \Phi^{\widetilde{H}}_{\epsilon, \tau} \right) 
\right]
+ \mathcal{O} ( \epsilon^{k + 2} ).
\end{align*}
Finally, we apply Corollary \ref{cor:density_of_states} to replace the modified normalization 
on the LHS with the true normalization and the necessary correction,
\begin{align*}
\frac{
\int_{M} f \left( h \circ \Phi^{\widetilde{H}}_{\epsilon, \tau} \right) e^{-H} \Omega
}{ \int_{M} e^{- H } \Omega }
\left( 1 + \epsilon^{k} \, \mathbb{E}_{\varpi_{\widetilde{H}}} \! \left[ \mathrm{d} G \! \left( \vec{v} \right)
+ G \left( \frac{\partial v^{i} }{ \partial q^{i} } + \frac{\partial v_{i} }{ \partial p_{i} } \right) \right] \right)
&= \\
& \hspace{-70mm}
\mathbb{E}_{\varpi_{\widetilde{H}}} \! \left[ \left( f \circ \Phi^{\widetilde{H}}_{\epsilon, \tau} \right) h \right]
+ \epsilon^{k} \,
\mathbb{E}_{\varpi_{\widetilde{H}}} \! \left[ \left( f \circ \Phi^{\widetilde{H}}_{\epsilon, \tau} \right) h \, 
\left( G \circ \Phi^{\widetilde{H}}_{\epsilon, \tau} \right) \right]
+ \mathcal{O} ( \epsilon^{k + 2} )
\\
\mathbb{E}_{\varpi_{H}} \! \left[ f \left( h \circ \Phi^{\widetilde{H}}_{\epsilon, \tau} \right) \right]
+ \epsilon^{k} \, 
\mathbb{E}_{\varpi_{H}} \! \left[ f \left( h \circ \Phi^{\widetilde{H}}_{\epsilon, \tau} \right) \right]
\mathbb{E}_{\varpi_{\widetilde{H}}} \! \left[ \mathrm{d} G \! \left( \vec{v} \right)
+ G \left( \frac{\partial v^{i} }{ \partial q^{i} } + \frac{\partial v_{i} }{ \partial p_{i} } \right) \right]
&= \\
& \hspace{-70mm}
\mathbb{E}_{\varpi_{\widetilde{H}}} \! \left[ \left( f \circ \Phi^{\widetilde{H}}_{\epsilon, \tau} \right) h \right]
+ \epsilon^{k} \,
\mathbb{E}_{\varpi_{\widetilde{H}}} \! \left[ \left( f \circ \Phi^{\widetilde{H}}_{\epsilon, \tau} \right) h \, 
\left( G \circ \Phi^{\widetilde{H}}_{\epsilon, \tau} \right) \right]
+ \mathcal{O} ( \epsilon^{k + 2} ),
\end{align*}
or
\begin{align*}
\mathbb{E}_{\varpi_{H}} \! \left[ f \left( h \circ \Phi^{\widetilde{H}}_{\epsilon, \tau} \right) \right] 
&= 
\mathbb{E}_{\varpi_{\widetilde{H}}} \! \left[ \left( f \circ \Phi^{\widetilde{H}}_{\epsilon, \tau} \right) h \right]
\\
& \quad + \epsilon^{k} \Big( 
- \mathbb{E}_{\varpi_{H}} \! \left[ f \left( h \circ \Phi^{\widetilde{H}}_{\epsilon, \tau} \right) \right]
\mathbb{E}_{\varpi_{\widetilde{H}}} \! \left[ \mathrm{d} G \! \left( \vec{v} \right)
+ G \left( \frac{\partial v^{i} }{ \partial q^{i} } + \frac{\partial v_{i} }{ \partial p_{i} } \right) \right]
\\
& \quad\quad\quad\; + \mathbb{E}_{\varpi_{\widetilde{H}}} \! \left[ 
\left( f \circ \Phi^{\widetilde{H}}_{\epsilon, \tau} \right) h \, \left( G \circ \Phi^{\widetilde{H}}_{\epsilon, \tau} \right) 
\right] \Big)
\\
& \quad + \mathcal{O} ( \epsilon^{k + 2} ),
\end{align*}
as desired.

\end{proof}

Now can compute the scaling of the moments of the Hamiltonian error, $\Delta_{\epsilon}$. 

\newtheorem*{lem:moment_scaling}{Lemma \ref{lem:moment_scaling}}
\begin{lem:moment_scaling}

Let $\left( M, \omega, H\right)$ be a Hamiltonian system and consider
a $k$-th order symmetric symmetric symplectic integrator with the corresponding modified Hamiltonian
$\widetilde{H} = H + \epsilon^{k} \, G + \mathcal{O} ( \epsilon^{k + 2} )$
and Metropolis proposal $\Phi^{\widetilde{H}}_{\epsilon, \tau}$. 
If the integrator is topologically stable and the asymptotic error is negligible, 
then the moments of the Hamiltonian error with respect to the unit canonical distribution
scale as
\begin{equation*}
\mathbb{E}_{\varpi_{H}} \! \left[ \left( \Delta_{\epsilon} \right)^{n} \right] \propto
\left\{
\begin{array}{rr}
\epsilon^{ k ( n + 1 )}, & n \, \mathrm{odd} \;\, \\
\epsilon^{k n} \;\,, & n \, \mathrm{even}
\end{array} 
\right. .
\end{equation*}
to leading-order in $\epsilon$.

\end{lem:moment_scaling}

\begin{proof}

First we explicitly introduce the modified Hamiltonian into error,
\begin{align*}
\mathbb{E}_{\varpi_{H}} \left[ \left( \Delta_{\epsilon} \right)^{n} \right]
&= 
\mathbb{E}_{\varpi_{H}} \left[
\left( H - H \circ \Phi^{\widetilde{H}}_{\epsilon, \tau} \right)^{n}
\right]
\\
&=
\mathbb{E}_{\varpi_{H}} \left[
\left( - \left( \widetilde{H} - H \right) + \widetilde{H} - H \circ \Phi^{\widetilde{H}}_{\epsilon, \tau} \right)^{n}
\right]
\end{align*}
Because $\widetilde{H}$ is invariant to the flow of the integrator we can drag the second
$\widetilde{H}$ along the flow to give
\begin{align*}
\mathbb{E}_{\varpi_{H}} \left[ \left( \Delta_{\epsilon} \right)^{n} \right]
&=
\mathbb{E}_{\varpi_{H}} \left[
\left( - \left( \widetilde{H} - H \right) + \widetilde{H} - H \circ \Phi^{\widetilde{H}}_{\epsilon, \tau}  \right)^{n}
\right]
\\
&=
\mathbb{E}_{\varpi_{H}} \left[
\left( - \left( \widetilde{H} - H \right) + \widetilde{H} \circ \Phi^{\widetilde{H}}_{\epsilon, \tau}  
- H \circ \Phi^{\widetilde{H}}_{\epsilon, \tau}  \right)^{n}
\right]
\\
&=
\mathbb{E}_{\varpi_{H}} \left[
\left( - \left( \widetilde{H} - H \right) 
+ \left( \widetilde{H} - H \right) \circ \Phi^{\widetilde{H}}_{\epsilon, \tau}  \right)^{n}
\right]
\\
&=
\epsilon^{k n} \, \mathbb{E}_{\varpi_{H}} \left[
\left( - G + G \circ \Phi^{\widetilde{H}}_{\epsilon, \tau}  \right)^{n}
\right]
+ \mathcal{O} ( \epsilon^{(k + 2) n} ).
\end{align*}

Now we expand the binomial,
\begin{align*}
\mathbb{E}_{\varpi_{H}} \left[ \left( \Delta_{\epsilon} \right)^{n} \right]
&=
\epsilon^{k n} \, \mathbb{E}_{\varpi_{H}} \left[
\left( - G + G \circ \Phi^{\widetilde{H}}_{\epsilon, \tau}  \right)^{n}
\right]
+ \mathcal{O} ( \epsilon^{(k + 2) n} )
\\
&=
\epsilon^{k n} \, \mathbb{E}_{\varpi_{H}} \left[
\sum_{j = 0}^{n} \binom{n}{k} 
\left( G \circ \Phi^{\widetilde{H}}_{\epsilon, \tau}  \right)^{j} \left( - G \right)^{n - j}
\right]
+ \mathcal{O} ( \epsilon^{(k + 2) n} ).
\end{align*}

For odd $n$ there are an even number of terms in the expansion which pair up as
\begin{align*}
\mathbb{E}_{\varpi_{H}} \left[ \left( \Delta_{\epsilon} \right)^{n} \right]
&=
\epsilon^{k n} \, \sum_{j = 0}^{\frac{n-1}{2}} \binom{n}{j} \left( - 1 \right)^{n - j}
\left(
\mathbb{E}_{\varpi_{H}} \left[
\left( G \circ \Phi^{\widetilde{H}}_{\epsilon, \tau}  \right)^{j} \left( G \right)^{n - j} 
\right]
-
\mathbb{E}_{\varpi_{H}} \left[
\left( G \circ \Phi^{\widetilde{H}}_{\epsilon, \tau}  \right)^{n - j} \left( G \right)^{n}
\right]
\right)
\\
& \quad + \mathcal{O} ( \epsilon^{(k + 2) n} ),
\end{align*}
or appealing to Lemma \ref{lem:correlation_by_parts},
\begin{align*}
\mathbb{E}_{\varpi_{H}} \left[ \left( \Delta_{\epsilon} \right)^{n} \right]
&=
\epsilon^{k n} \, \sum_{j = 0}^{\frac{n-1}{2}} \binom{n}{j} \left( - 1 \right)^{n - j}
\left(
\mathbb{E}_{\varpi_{H}} \left[
\left( G \circ \Phi^{\widetilde{H}}_{\epsilon, \tau}  \right)^{j} \left( G \right)^{n - j} 
\right]
-
\mathbb{E}_{\varpi_{\widetilde{H}}} \left[
\left( G \circ \Phi^{\widetilde{H}}_{\epsilon, \tau}  \right)^{j} \left( G \right)^{n - j} 
\right] \right)
\\
& \quad +
\epsilon^{k n} \, \sum_{j = 0}^{\frac{n-1}{2}} \binom{n}{j} \left( - 1 \right)^{n - j}
\left(
\epsilon^{k} \, \mathbb{E}_{\varpi_{\widetilde{H}}} \left[ \textrm{boundary term} \right]
+ \mathcal{O} ( \epsilon^{k + 2} )
\right)
+ \mathcal{O} \! \left( \epsilon^{(k + 2) n} \right)
\\
&=
\epsilon^{k n} \, \sum_{j = 0}^{\frac{n-1}{2}} \binom{n}{j} \left( - 1 \right)^{n - j}
\left(
\mathbb{E}_{\varpi_{H}} \left[
\left( G \circ \Phi^{\widetilde{H}}_{\epsilon, \tau}  \right)^{j} \left( G \right)^{n - j} 
\right]
-
\mathbb{E}_{\varpi_{\widetilde{H}}} \left[
\left( G \circ \Phi^{\widetilde{H}}_{\epsilon, \tau}  \right)^{j} \left( G \right)^{n - j} 
\right] \right)
\\
& \quad 
+ \mathcal{O} \! \left( \epsilon^{ k (n + 1)} \right)
+ \mathcal{O} \! \left( \epsilon^{ k (n + 2)} \right)
+ \mathcal{O} \! \left( \epsilon^{(k + 2) n} \right),
\end{align*}
and then finally Theorem \ref{thm:canonical},
\begin{align*}
\mathbb{E}_{\varpi_{H}} \left[ \left( \Delta_{\epsilon} \right)^{n} \right]
&= 
\mathcal{O} \! \left( \epsilon^{k (n + 1)} \right) 
+ \mathcal{O} \! \left( \epsilon^{ k (n + 1)} \right) 
+ \mathcal{O} \! \left( \epsilon^{ k (n + 2) } \right) 
+ \mathcal{O} \! \left( \epsilon^{(k + 2) n} \right)
\\
&= \mathcal{O} \! \left( \epsilon^{ k (n + 1) } \right).
\end{align*}

Even powers do not benefit from a similar cancelation and we're left
with the nominal scaling which gives
\begin{equation*}
\mathbb{E}_{\varpi_{H}} \! \left[ \left( \Delta_{\epsilon} \right)^{n} \right] \propto
\left\{
\begin{array}{rr}
\epsilon^{k (n + 1)}, & n \, \mathrm{odd} \;\, \\
\epsilon^{k n} \;\,, & n \, \mathrm{even}
\end{array} 
\right. ,
\end{equation*}
as desired.

\end{proof}

Given the moments the scaling of the cumulants immediately follows.

\newtheorem*{lem:cumulant_scaling}{Lemma \ref{lem:cumulant_scaling}}
\begin{lem:cumulant_scaling}

Let $\left( M, \omega, H\right)$ be a Hamiltonian system and consider
a $k$-th order symmetric symplectic integrator with the corresponding modified Hamiltonian
$\widetilde{H} = H + \epsilon^{k} \, G + \mathcal{O} ( \epsilon^{k + 2} )$
and Metropolis proposal $\Phi^{\widetilde{H}}_{\epsilon, \tau}$.  
If the  integrator is topologically stable and the asymptotic error is negligible, 
then the cumulants of the Hamiltonian error with respect to the unit canonical distribution scale 
as
\begin{equation*}
\kappa_{n} \! \left( \Delta_{\epsilon} \right) \propto
\left\{
\begin{array}{rr}
\epsilon^{k (n + 1) }, & n \, \mathrm{odd} \;\, \\
\epsilon^{k n} \;\,, & n \, \mathrm{even}
\end{array} 
\right. .
\end{equation*}
to leading-order in $\epsilon$.

\end{lem:cumulant_scaling}

\begin{proof}

Cumulants can be constructed from the moments via the recursion relation
\begin{equation*}
\kappa_{n} = \mu_{n} - \sum_{m = 1}^{n - 1} \binom{n - 1}{m - 1} \, \kappa_{m} \, \mu_{n - m},
\end{equation*}
and we use this relationship to proceed inductively.

Provided that the scaling holds up to $\kappa_{n - 1}$, then if $n$ is even
each term in the sum is a product of terms with like parity whereas
if $n$ is odd then each term is the sum is a product of terms with odd parity.
Consequently each term scales with $\epsilon^{k m}, \, m \ge n$ and to leading
order $\kappa_{n} \propto \epsilon^{k n}$.  The base case is confirmed immediately
as $\kappa_{1} = \mu_{1}$.

\end{proof}

Moreover, the symplectic structure provides a global constraint on the cumulants.

\newtheorem*{lem:global_constraint}{Lemma \ref{lem:global_constraint}}
\begin{lem:global_constraint}

Let $\left( M, \omega, H \right)$ be a Hamiltonian system and consider
a $k$-th order symmetric symplectic integrator with the corresponding modified Hamiltonian
$\widetilde{H} = H + \epsilon^{k} \, G + \mathcal{O} ( \epsilon^{k + 2} )$
and Metropolis proposal $\Phi^{\widetilde{H}}_{\epsilon, \tau}$. 
If the integrator is topologically stable and the asymptotic error is negligible, 
then the cumulant generating function of the Hamiltonian error vanishes
\begin{equation*}
\log \mathbb{E}_{\varpi_{H}} \! \left[ e^{ \Delta_{\epsilon}} \right]  = 0.
\end{equation*}
for the canonical distribution with $\beta = 1$,
\begin{equation*}
\varpi_{H} = \frac{ e^{-H} \Omega }{ \int_{M} e^{-H} \Omega }.
\end{equation*}

\end{lem:global_constraint}

\begin{proof}

By definition
\begin{align*}
\mathbb{E}_{\varpi_{H}} \left[ e^{ \Delta_{\epsilon}} \right] 
&= 
\frac{ \int_{M} e^{ \Delta_{\epsilon}} e^{-H} \Omega }
{ \int_{M} e^{-H} \Omega }
\\
&= 
\frac{ \int_{M} e^{ H - H \circ \Phi^{\widetilde{H}}_{\epsilon, \tau} } e^{-H} \Omega }
{ \int_{M} e^{-H} \Omega }
\\
&= 
\frac{ \int_{M} e^{- H \circ \Phi^{\widetilde{H}}_{\epsilon, \tau} } \Omega }
{ \int_{M} e^{-H} \Omega }.
\end{align*}
Pushing back the numerator against $\Phi^{\widetilde{H}}_{\epsilon, \tau}$ finally gives
\begin{align*}
\mathbb{E}_{\varpi_{H}} \left[ e^{ \Delta_{\epsilon}} \right] 
&= 
\frac{ \int_{M} \left( \Phi^{\widetilde{H}}_{\epsilon, -\tau} \right)^{*}
\left( e^{- H \circ \Phi^{\widetilde{H}}_{\epsilon, \tau} } \Omega \right) } 
{ \int_{M} e^{-H} \Omega }
\\
&= 
\frac{ \int_{M} e^{- H \circ \Phi^{\widetilde{H}}_{\epsilon, \tau} \circ \Phi^{\widetilde{H}}_{\epsilon, -\tau} } \Omega }
{ \int_{M} e^{-H} \Omega }
\\
&= 
\frac{ \int_{M} e^{- H } \Omega }
{ \int_{M} e^{-H} \Omega }
\\
&= 1,
\end{align*}
or
\begin{equation*}
\log \mathbb{E}_{\varpi_{H}} \left[ e^{ \Delta_{\epsilon}} \right] = 0,
\end{equation*}
as desired.

\end{proof}

In order to compute expectations of functions of the Hamiltonian error 
in the general case we appeal to the Gram-Charlier expansion \citep{BarndorffEtAl:1989}.

\newtheorem*{thm:smooth_expectations}{Theorem \ref{thm:smooth_expectations}}
\begin{thm:smooth_expectations}

Let $\left( M, \omega, H\right)$ be a Hamiltonian system and consider
a $k$-th order symmetric symplectic integrator with the corresponding modified Hamiltonian
$\widetilde{H} = H + \epsilon^{k} \, G + \mathcal{O} ( \epsilon^{k + 2} )$
and Metropolis proposal $\Phi^{\widetilde{H}}_{\epsilon, \tau}$.  If the integrator 
is topologically stable and the asymptotic error is negligible, then the expectation of 
any smooth function of the Hamiltonian error is given by
\begin{align*}
\mathbb{E}_{\varpi_{H} } \! \left[ f \! \left( \Delta_{\epsilon} \right) \right] 
&=
\int^{\infty}_{-\infty} \mathrm{d} \Delta_{\epsilon} \, 
\mathcal{N} \! \left( \Delta_{\epsilon} | 
- \frac{1}{2} \alpha \, \epsilon^{2k}, \alpha \, \epsilon^{2k} \right) 
f \! \left( \Delta_{\epsilon} \right)
+ \mathcal{O} ( \epsilon^{2 (k + 2)} ),
\end{align*}
for some $\alpha \in \mathbb{R}$.

\end{thm:smooth_expectations}

\begin{proof}

The Gram-Charlier series defines an expansion of a target density $\varpi \! \left( x \right)$ in
terms of derivatives of a reference Gaussian density,
\begin{equation*}
\varpi \! \left( x \right) = \exp \! \left[ \sum_{n = 1}^{\infty} \frac{ \kappa_{n} - \gamma_{n} }{ n! } 
\left( - \frac{\partial}{\partial x} \right)^{n} \right] 
\mathcal{N} \! \left( x | \mu, \sigma^{2} \right),
\end{equation*}
where $\kappa_{n}$ are the cumulants of $\pi$ and $\gamma_{n}$ are the cumulants of the Gaussian.  
By matching the Gaussian to the first two moments of $\varpi$, the first two terms in the sum vanish leaving
\begin{equation*}
\varpi \! \left( x \right) = 
\exp \! \left[ \sum_{n = 3}^{\infty} \left( - 1 \right)^{n} \frac{ \kappa_{n} }{ n! } 
\left( \frac{\partial}{\partial x} \right)^{n} \right] 
\mathcal{N} \! \left( x | \kappa_{1}, \kappa_{2} \right).
\end{equation*}

Given this expansion, expectations with respect to $x$ can be written as
\begin{align*}
\mathbb{E}_{\varpi_{H}} \! \left[ f \right] 
&= \int^{\infty}_{-\infty} \mathrm{d}x \, \varpi \! \left( x \right) f \! \left( x \right) \\
&= \int^{\infty}_{-\infty} \mathrm{d}x \, 
\exp \! \left[ \sum_{n = 3}^{\infty} \left( - 1 \right)^{n} \frac{ \kappa_{n} }{ n! } 
\left( - \frac{\partial}{\partial x} \right)^{n} \right] 
\mathcal{N} \! \left( x | \kappa_{1}, \kappa_{2} \right) \cdot f \! \left( x \right),
\end{align*}
or, upon changing variables,
\begin{align*}
\mathbb{E}_{\varpi_{H}} \! \left[ f \right] 
&=
\int^{\infty}_{-\infty} \mathrm{d}y \, 
\exp \! \left[ \sum_{n = 3}^{\infty} \left( -1 \right)^{n} \frac{ \kappa_{n} }{ n! \, \kappa_{2}^{n/2} } 
\left( - \frac{\partial}{\partial y} \right)^{n} \right] 
\mathcal{N} \! \left( y | 0, 1 \right) \cdot f \! \left( \sqrt{\kappa_{2}} y + \kappa_{1} \right) \\
&=
\int^{\infty}_{-\infty} \mathrm{d}y \, 
\exp \! \left[ \sum_{n = 3}^{\infty} \left( -1 \right)^{n} \frac{ \kappa_{n} }{ n! \, \kappa_{2}^{n/2} } 
\left( - \frac{\partial}{\partial y} \right)^{n} \right] 
\phi \! \left( n \right) \cdot f \! \left( \sqrt{\kappa_{2}} y + \kappa_{1} \right).
\end{align*}

Given the internal summation, expanding the exponential is no easy task.  The action 
of each term, however, is relatively easy to deduce: because there is no $y$ dependence 
in the exponential,  any term in the expansion will reduce to
\begin{align*}
I = C_{n}
\int^{\infty}_{-\infty} \mathrm{d}y \, 
\frac{\partial^{n} }{\partial y^{n} }
\phi \! \left( y \right) \cdot f \! \left( \sqrt{\kappa_{2}} y + \kappa_{1} \right),
\end{align*}
with $C_{n}$ some product of coefficients, $\left( -1 \right)^{m} \, \kappa_{m} \kappa_{2}^{-m/2} / m!$,
whose order sums to $n$. Integrating this term by parts yields
\begin{align*}
I &= 
C_{n} \left. \frac{\partial^{n - 1} }{\partial y^{n - 1} }
\phi \! \left( y \right) \cdot f \! \left( \sqrt{\kappa_{2}} y + \kappa_{1} \right) \right|^{\infty}_{-\infty} \\
& 
- C_{n} \int^{\infty}_{-\infty} \mathrm{d}y \, 
\frac{\partial^{n - 1} }{\partial y^{n - 1} }
\phi \! \left( y \right) \cdot \sigma f^{\left(1\right)} \! \left( \sqrt{\kappa_{2}} y + \kappa_{1} \right) \\
&=
- C_{n} \, \sqrt{\kappa_{2}} \int^{\infty}_{-\infty} \mathrm{d}y \, 
\frac{\partial^{n - 1} }{\partial y^{n - 1} }
\phi \! \left( y \right) \cdot  f^{\left(1\right)} \! \left( \sqrt{\kappa_{2}} y + \kappa_{1} \right).
\end{align*}
Upon repeated integration by parts this eventually reduces to
\begin{align*}
I 
&= C_{n} \left( -1 \right)^{n} \kappa_{2}^{n/2} \int^{\infty}_{-\infty} \mathrm{d}y \, 
\phi \! \left( n \right) f^{\left(n \right)} \! \left( \sqrt{\kappa_{2}} y + \kappa_{1} \right).
\end{align*}

Now let $x = \Delta_{\epsilon}$ with $\kappa_{1} = \mathbb{E}_{\varpi_{H}} \! \left[ \Delta_{\epsilon} \right]$
and $\kappa_{2} = \mathrm{Var}_{\pi} \! \left[ \Delta_{\epsilon} \right]$.  From Lemma 
\ref{lem:cumulant_scaling} we have
\begin{equation*}
\kappa_{n} \propto \left\{ 
\begin{array}{lr}
\kappa_{2}^{(n + 1) / 2}, & n \, \mathrm{odd} \\
\kappa_{2}^{n}, & n \, \mathrm{even}
\end{array} \right. ,
\end{equation*}
or
\begin{equation*}
\frac{ \kappa_{n} }{ n! \, \kappa_{2}^{n/2} } \propto
\left\{ \begin{array}{lr}
\sqrt{\kappa_{2}}, & n \, \mathrm{odd} \\
1, & n \, \mathrm{even}
\end{array} \right. .
\end{equation*}
Consequently $C_{n}$ will only introduce addition factors of $\kappa_{2}$,
\begin{equation*}
I \propto \kappa_{2}^{m/2} \int^{\infty}_{-\infty} \mathrm{d}y \, 
\phi \! \left( y \right) f^{\left(n \right)} \! \left( \sqrt{\kappa_{2}} y + \kappa_{1} \right), \, m \ge n.
\end{equation*}

If $f$ is smooth then the derivatives $f^{\left(n \right)} \! \left( \sqrt{\kappa_{2}} y + \kappa_{1} \right)$
will not introduce any addition factors of $\kappa_{1}$ nor $\kappa_{2}$ as we have already
incorporated the contributions from the Jacobian.  Consequently the  first contribution from the expansion 
beyond unity will be given by
\begin{equation*}
I \propto \kappa_{2}^{2} \int^{\infty}_{-\infty} \mathrm{d}y \, 
\phi \! \left( y \right) f^{\left(3 \right)} \! \left( \sqrt{\kappa_{2}} y + \kappa_{1} \right),
\end{equation*}
and to leading-order the expectation becomes
\begin{align*}
\mathbb{E}_{\pi } \! \left[ f \! \left( \Delta_{\epsilon} \right) \right] 
&= 
\int^{\infty}_{-\infty} \mathrm{d}y \, 
\exp \! \left[ \sum_{n = 3}^{\infty} \left( -1 \right)^{n} \frac{ \kappa_{n} }{ n! \, \kappa_{2}^{n/2} } 
\left( \frac{\partial}{\partial y} \right)^{n} \right] 
\phi \! \left( y \right) \cdot f \! \left( \sqrt{\kappa_{2}} y + \kappa_{1} \right) \\
&=
\int^{\infty}_{-\infty} \mathrm{d}y \, 
\phi \! \left( y \right) f \! \left( \sqrt{\kappa_{2}} y + \kappa_{1} \right)
+ \mathcal{O} \! \left( \kappa_{2}^{2} \right) \\
&=
\int^{\infty}_{-\infty} \mathrm{d} \Delta_{\epsilon} \, 
\mathcal{N} \! \left( \Delta_{\epsilon} | \kappa_{1}, \kappa_{2} \right) f \! \left( \Delta_{\epsilon} \right)
+ \mathcal{O} \! \left( \kappa_{2}^{2} \right),
\end{align*}
or substituting the results of Corollary \ref{cor:leading_cumulants},
\begin{equation*}
\kappa_{1} = - \frac{1}{2} \kappa_{2} = - \frac{1}{2} \alpha \, \epsilon^{2k},
\end{equation*}
\begin{align*}
\mathbb{E}_{\pi } \! \left[ f \! \left( \Delta_{\epsilon} \right) \right] 
&=
\int^{\infty}_{-\infty} \mathrm{d} \Delta_{\epsilon} \, 
\mathcal{N} \! \left( \Delta_{\epsilon} | 
- \frac{1}{2} \alpha \, \epsilon^{2k}, \alpha \, \epsilon^{2k}\right) 
f \! \left( \Delta_{\epsilon} \right)
+ \mathcal{O} ( \epsilon^{2 (k + 2)} ),
\end{align*}
as desired.

\end{proof}

At leading-order in $\epsilon$, the expectation of any \textit{smooth} function of $\Delta_{\epsilon}$ 
becomes a straightforward Gaussian integral, equivalent to the infinite independently distributed limit.  

Unfortunately the expectation in which we are mainly interested, the Metropolis acceptance probability, 
is not smooth because of a cusp at $\Delta_{\epsilon} = 0$.  The cusp introduces non-trivial boundary 
terms that then induce additional leading-order contributions to the expectations beyond those found 
in the smooth case.

\newtheorem*{thm:acceptance_expectation}{Theorem \ref{thm:acceptance_expectation}}
\begin{thm:acceptance_expectation}

Let $\left( M, \omega, H\right)$ be a Hamiltonian system and consider
a $k$-th order symmetric symplectic integrator with the corresponding modified Hamiltonian
$\widetilde{H} = H + \epsilon^{k} \, G + \mathcal{O} ( \epsilon^{k + 2} )$
and Metropolis proposal $\Phi^{\widetilde{H}}_{\epsilon, \tau}$.  If the integrator 
is topologically stable and the asymptotic error is negligible, then the expectation of 
the Metropolis acceptance probability is given by
\begin{align*}
\mathbb{E}_{\varpi_{H}} \! \left[ a \! \left( \Delta_{\epsilon} \right) \right] 
=&
\int^{\infty}_{-\infty} \mathrm{d} \Delta_{\epsilon} \, 
\mathcal{N} \! \left( \Delta_{\epsilon} | 
- \frac{1}{2} \alpha \, \epsilon^{2 k}, \alpha \, \epsilon^{2 k} \right) 
a \! \left( \Delta_{\epsilon} \right)
+ \mathcal{O} ( \epsilon^{k} ),
\end{align*}
for some $\alpha \in \mathbb{R}$.

\end{thm:acceptance_expectation}

\begin{proof}

In order to understand the effect of the cusp it is easiest to proceed as with Theorem
\ref{thm:smooth_expectations} up until each term is integrating by parts,
\begin{align*}
I = C_{n} \int^{\infty}_{-\infty} \mathrm{d}y \, 
\frac{\partial^{n} }{\partial y^{n} }
\phi \! \left( y \right) \cdot \min \! \left(1, e^{- \sqrt{ \kappa_{2} } y - \kappa_{1} } \right).
\end{align*}
On the first integration by parts the boundary term vanishes as before,
\begin{align*}
I =& \,\,
C_{n} \left. \frac{\partial^{n - 1} }{\partial y^{n - 1} } \phi \! \left( y \right) \cdot
\min \! \left(1, e^{- \sqrt{ \kappa_{2} } y - \kappa_{1} } \right) \right|^{\infty}_{-\infty} \\
&
- C_{n} \int^{\infty}_{-\infty} \mathrm{d}y \, 
\frac{\partial^{n - 1} }{\partial y^{n - 1} } \phi \! \left( y \right) \cdot
\left( - \sqrt{ \kappa_{2} } \right) \, e^{- \sqrt{ \kappa_{2} } y - \kappa_{1} } \, 
\mathbb{I} \left( - \sqrt{ \kappa_{2} } y - \kappa_{1} \right) \\
=& \,\,
C_{n} \, \sqrt{ \kappa_{2} } \int^{\infty}_{- \kappa_{1} / \sqrt{ \kappa_{2} } } \mathrm{d}y \, 
\frac{\partial^{n - 1} }{\partial y^{n - 1} } \phi \! \left( y \right) \cdot
e^{- \sqrt{ \kappa_{2} } - \kappa_{1} }.
\end{align*}
Continued applications, however, introduce nontrivial boundary contributions,
\begin{align*}
I / C_{n} =& \,\,
\sqrt{ \kappa_{2} } \left. \frac{\partial^{n - 2} }{\partial y^{n - 2} }
\phi \! \left( y \right) \cdot e^{- \sqrt{ \kappa_{2} } y - \kappa_{1} } 
\right|^{\infty}_{- \kappa_{1} / \sqrt{ \kappa_{2} }} \\
&
- \sqrt{ \kappa_{2} } \int^{\infty}_{- \kappa_{1} / \sqrt{ \kappa_{2} } } \mathrm{d}y \, 
\frac{\partial^{n - 2} }{\partial y^{n - 2} }
\phi \! \left( y \right) \cdot \left( - \sqrt{ \kappa_{2} } \right) e^{- \sqrt{ \kappa_{2} } y - \kappa_{1} } \\
=& \,\,
- \sqrt{ \kappa_{2} } \frac{\partial^{n - 2} }{\partial y^{n - 2} }
\phi \! \left( - \frac{ \kappa_{1} }{ \sqrt{ \kappa_{2} } } \right) \\
&
\kappa_{2} \int^{\infty}_{- \kappa_{1} / \sqrt{ \kappa_{2} } } \mathrm{d}y \, 
\frac{\partial^{n - 2} }{\partial y^{n - 2} }
\phi \! \left( y \right) \cdot e^{- \sqrt{ \kappa_{2} } y - \kappa_{1} } \\
\end{align*} 
and, eventually,
\begin{align*}
I / C_{n} =& 
- \sum_{k = 1}^{n - 1} \kappa_{2}^{ k/2 } \frac{\partial^{n - k - 1} }{\partial y^{n - k - 1} }
\phi \! \left( - \frac{ \kappa_{1} }{ \sqrt{ \kappa_{2} } } \right) \\
&+ \kappa_{2}^{n/2} \int^{\infty}_{- \kappa_{1} / \sqrt{ \kappa_{2} } } \mathrm{d}y \, 
\phi \! \left( y \right) e^{- \sqrt{ \kappa_{2} } y - \kappa_{1} } \\
=&
- \phi \! \left( - \frac{ \kappa_{1} }{ \sqrt{ \kappa_{2} } } \right) \sum_{k = 1}^{n - 1} 
\left( -1 \right)^{n - k - 1} \kappa_{2}^{k/2} H_{n - k - 1} \! \left( - \frac{ \kappa_{1} }{ \sqrt{ \kappa_{2} } } \right) \\
&+ \kappa_{2}^{n/2} \int^{\infty}_{- \kappa_{1} / \sqrt{ \kappa_{2} } } \mathrm{d}y \, 
\phi \! \left( y \right) e^{- \sqrt{ \kappa_{2} } y - \kappa_{1}} \\
=&
\phi \! \left( - \frac{ \kappa_{1} }{ \sqrt{ \kappa_{2} } } \right) \sum_{k = 1}^{n - 1} 
\left( -1 \right)^{n - k} \kappa_{2}^{k/2} H_{n - k - 1} \! \left( - \frac{ \kappa_{1} }{ \sqrt{ \kappa_{2} } } \right) \\
&+ \kappa_{2}^{n/2} \int^{\infty}_{- \kappa_{1} / \sqrt{ \kappa_{2} } } \mathrm{d}y \, 
\phi \! \left( y \right) e^{- \sqrt{ \kappa_{2} } y - \kappa_{1} } \\
=&
\phi \! \left( \frac{ \sqrt{ \kappa_{2} } }{2} \right) \sum_{k = 1}^{n - 1} 
\left( -1 \right)^{n - k} \kappa_{2}^{k/2} H_{n - k - 1} \! \left( \frac{ \sqrt{ \kappa_{2} } }{2} \right) \\
&+ \kappa_{2}^{n/2} \int^{\infty}_{ \frac{ \sqrt{ \kappa_{2} } }{2} } \mathrm{d}y \, 
\phi \! \left( y \right) e^{- \sqrt{ \kappa_{2} } y + \kappa_{2} / 2 }.
\end{align*}

The second term is exactly the result we would have if the acceptance
probability were smooth, with the contributions from the cusp isolated to the first
term.  In particular, the Hermite polynomials introduce terms proportional
to $\sqrt{ \kappa_{2} }$ and $\kappa_{2}$ so that at best we have
\begin{align*}
\mathbb{E}_{\pi } \! \left[ a \! \left( \Delta_{\epsilon} \right) \right] 
=&
\int^{\infty}_{-\infty} \mathrm{d} \Delta_{\epsilon} \, 
\mathcal{N} \! \left( \Delta_{\epsilon} | 
- \frac{1}{2} \alpha \, \epsilon^{2k}, \alpha \, \epsilon^{2k}\right) 
a \! \left( \Delta_{\epsilon} \right)
\\
&+ D_{1} \, \epsilon^{k} + D_{2} \, \epsilon^{2 k}
+ \mathcal{O} ( \epsilon^{2 (k + 2)} ),
\end{align*}
or
\begin{align*}
\mathbb{E}_{\pi } \! \left[ a \! \left( \Delta_{\epsilon} \right) \right] 
=&
\int^{\infty}_{-\infty} \mathrm{d} \Delta_{\epsilon} \, 
\mathcal{N} \! \left( \Delta_{\epsilon} |
- \frac{1}{2} \alpha \, \epsilon^{2k}, \alpha \, \epsilon^{2k}\right) 
a \! \left( \Delta_{\epsilon} \right) + \mathcal{O} ( \epsilon^{k} ),
\end{align*}
as desired.

\end{proof}

Although the expectation of the Metropolis acceptance probability is not equivalent
to the infinite independently distributed limit, the deviations are isolated into two terms, 
$D_{1}$ and $D_{2}$, which admit further study.  Indeed, empirically these terms appear 
to be small indicating that there may be a means of constraining their values in general
and improving this result.

\subsection{Approximating the Average Hamiltonian Error}
\label{apx:ave_error}

Taking $n = 1$ in Lemma \ref{lem:moment_scaling} gives the full leading-order result 
for the average Hamiltonian error,
\begin{align*}
\mathbb{E}_{\varpi_{H}} \left[ \Delta_{\epsilon} \right] 
&=
+ \epsilon^{2k} \,
\mathbb{E}_{\varpi_{\widetilde{H}}} \! \left[ 2 G \, \mathrm{d} G \! \left( \vec{v} \right) 
+ G^{2} \left( \frac{\partial v^{i} }{ \partial q^{i} } + \frac{\partial v_{i} }{ \partial p_{i} } \right) \right]
\\
& \quad - \epsilon^{2k} \,
\mathbb{E}_{\varpi_{\widetilde{H}}} \! \Big[ G \Big]
\mathbb{E}_{\varpi_{\widetilde{H}}} \! \left[ \mathrm{d} G \! \left( \vec{v} \right)
+ G \left( \frac{\partial v^{i} }{ \partial q^{i} } + \frac{\partial v_{i} }{ \partial p_{i} } \right) \right]
\\
& \quad + \epsilon^{2k} \, 
\mathbb{E}_{\varpi_{H}} \! \Big[ G \Big]
\mathbb{E}_{\varpi_{\widetilde{H}}} \! \left[ \mathrm{d} G \! \left( \vec{v} \right)
+ G \left( \frac{\partial v^{i} }{ \partial q^{i} } + \frac{\partial v_{i} }{ \partial p_{i} } \right) \right]
\\
& \quad - \epsilon^{2k} \, 
\mathbb{E}_{\varpi_{\widetilde{H}}} \! \left[ G \left( G \circ \Phi^{\widetilde{H}}_{\epsilon, \tau} \right) \right]
+ \mathcal{O} ( \epsilon^{k + 2} )
\\
&=
+ \epsilon^{2k} \,
\mathbb{E}_{\varpi_{\widetilde{H}}} \! \left[ 2 G \, \mathrm{d} G \! \left( \vec{v} \right) 
+ G^{2} \left( \frac{\partial v^{i} }{ \partial q^{i} } + \frac{\partial v_{i} }{ \partial p_{i} } \right) \right]
\\
& \quad + \epsilon^{2k} \, 
\left( \mathbb{E}_{\varpi_{H}} \! \Big[ G \Big] - \mathbb{E}_{\varpi_{\widetilde{H}}} \! \Big[ G \Big] \right)
\mathbb{E}_{\varpi_{\widetilde{H}}} \! \left[ \mathrm{d} G \! \left( \vec{v} \right)
+ G \left( \frac{\partial v^{i} }{ \partial q^{i} } + \frac{\partial v_{i} }{ \partial p_{i} } \right) \right]
\\
& \quad - \epsilon^{2k} \, 
\mathbb{E}_{\varpi_{\widetilde{H}}} \! \left[ G \left( G \circ \Phi^{\widetilde{H}}_{\epsilon, \tau} \right) \right]
+ D \, \epsilon^{k + 2}
+ \mathcal{O} ( \epsilon^{k + 4} ),
\end{align*}
where $D$ depends on the second-order contribution to the modified Hamiltonian.

To leading-order we can exchange the modified canonical expectations with canonical expectations,
$ \mathbb{E}_{\varpi_{\widetilde{H}}} [ \cdot ] \rightarrow \mathbb{E}_{\varpi_{H}} [ \cdot ]$, and the 
approximate flow with the exact flow, $\Phi^{\widetilde{H}}_{\epsilon, \tau} \rightarrow \Phi^{H}_{\epsilon, \tau}$, 
to give
\begin{align*}
\mathbb{E}_{\varpi_{H}} \left[ \Delta_{\epsilon} \right] 
&=
+ \epsilon^{2k} \,
\mathbb{E}_{\varpi_{H}} \! \left[ 2 G \, \mathrm{d} G \! \left( \vec{v} \right) 
+ G^{2} \left( \frac{\partial v^{i} }{ \partial q^{i} } + \frac{\partial v_{i} }{ \partial p_{i} } \right) \right]
\\
& \quad - \epsilon^{2k} \, 
\mathbb{E}_{\varpi_{H}} \! \left[ G \left( G \circ \Phi^{H}_{\tau} \right) \right]
+ D \, \epsilon^{k + 2} + \mathcal{O} ( \epsilon^{k + 4} ),
\end{align*}

In the Gaussian case the integrals can be computed analytically and provide
us with a means of validating Lemma \ref{lem:moment_scaling} against numerical 
experiments.  Here we consider second-order leapfrog integrators, encompassing both 
the explicit Stromer-Verlet integrator and the implicit midpoint integrator common in 
Hamiltonian Monte Carlo implementations.  Here
\begin{equation*}
G = \frac{1}{24} \left( 
2 \frac{ \partial H}{ \partial q^{i} } \frac{ \partial H}{ \partial q^{j} } 
\frac{ \partial^{2} H }{ \partial p_{i} \partial p_{j} } 
- \frac{ \partial H}{ \partial p_{i} } \frac{ \partial H}{ \partial p_{j} } 
\frac{ \partial^{2} H }{ \partial q^{i} \partial q^{j} } 
+ 2 \frac{ \partial H}{ \partial q^{i} } \frac{ \partial H}{ \partial p_{j} } 
\frac{ \partial^{2} H }{ \partial p_{i} \partial q^{j} } 
\right),
\end{equation*}
and all higher-order contributions to the modified Hamiltonian vanish
so that $D = 0$.

The Gaussian target distribution induces the Hamiltonian
\begin{equation*}
H = \frac{1}{2} p^{2} + \frac{1}{2} q^{2},
\end{equation*} 
with the sub-leading contribution to the modified Hamiltonian given by
\begin{equation*}
G = 
\frac{1}{24} \left( 2 q^{2} - p^{2} \right).
\end{equation*}

Computing the canonical expectation requires a transverse vector field, $\vec{v}$,  
and given the underlying Euclidean geometry with unit metric $\delta^{ij}$ an immediate 
choice is
\begin{align*}
\vec{v} 
=&
\, \gamma \left( 
 \delta_{ij} \, \frac{ \partial H}{ \partial q^{j} } \frac{ \partial }{ \partial q^{i} }, 
\delta^{ij} \, \frac{ \partial H}{ \partial p_{j} } \frac{ \partial }{ \partial p_{i} } 
\right) \\
=&
\, \gamma \left( 
q^{i} \frac{ \partial }{ \partial q^{i} }, 
p_{i} \frac{ \partial }{ \partial p_{i} } 
\right)
\end{align*}
with
\begin{equation*}
\gamma = \left( 2H \right)^{-1}.
\end{equation*}
This choice gives
\begin{equation*}
\frac{\partial v^{i} }{ \partial q^{i} } + \frac{\partial v_{i} }{ \partial p_{i} } = 0
\end{equation*}
and
\begin{equation*}
\mathrm{d} G \! \left( \vec{v} \right) = \frac{G}{H},
\end{equation*}
so that
\begin{align*}
\mathbb{E}_{\varpi_{H}} \left[ \Delta_{\epsilon} \right] 
&=
2 \epsilon^{4} \,
\mathbb{E}_{\varpi_{H}} \! \left[ \frac{G^{2}}{H}  \right]
- \epsilon^{4} \, 
\mathbb{E}_{\varpi_{H}} \! \left[ G \left( G \circ \Phi^{H}_{\tau} \right) \right]
+ \mathcal{O} \! \left( \epsilon^{6} \right),
\end{align*}
Finally we use the fact that the action of the exact flow is simply a rotation
in phase space,
\begin{align*}
 \Phi^{H}_{\tau} \left(q, p \right)
&= 
\left( \cos \tau \cdot q + \sin \tau \cdot p, - \sin \tau \cdot q + \cos \tau \cdot p \right)
\end{align*}

Together this gives
\begin{align*}
\mathbb{E}_{\varpi_{H}} \left[ \Delta_{\epsilon} \right] 
&=
2 \epsilon^{4} \,
\mathbb{E}_{\varpi_{H}} \! \left[ \frac{G^{2}}{H}  \right]
\\
& \quad - \epsilon^{4} \left( 
\cos^{2} \tau \, \mathbb{E}_{\varpi_{H}} \! \left[ G \frac{2 q^{2} - p^{2} }{24} \right]
+ 6 \sin \tau \cos \tau \, \mathbb{E}_{\varpi_{H}} \! \left[ G \frac{ q \, p }{24} \right]
\sin^{2} \tau \, \mathbb{E}_{\varpi_{H}} \! \left[ G \frac{2 p^{2} - q^{2} }{24} \right] \right)
\\
& \quad + \mathcal{O} \! \left( \epsilon^{6} \right)
\\
&=
\frac{11}{576} \epsilon^{4}
- \cos^{2} \tau \frac{11}{576} \epsilon^{4}
+ \sin^{2} \tau \frac{7}{576} \epsilon^{4}
+ \mathcal{O} \! \left( \epsilon^{6} \right)
\\
&=
\frac{9}{576} \epsilon^{4} \left(1 - \cos^{2} \tau + \sin^{2} \tau \right)
+ \mathcal{O} \! \left( \epsilon^{6} \right)
\\
&=
\frac{1}{64} \epsilon^{4} \left(1 - \cos 2 \tau \right)
+ \mathcal{O} \! \left( \epsilon^{6} \right).
\end{align*}

\bibliographystyle{imsart-nameyear}
\renewcommand{\bdoi}[1]{\textsc{doi}: \href{http://dx.doi.org/#1}{\nolinkurl{#1}}}
\renewcommand{\arxiv}[1]{\textsc{arxiv}: \href{http://arxiv.org/abs/#1}{\nolinkurl{#1}}}
\bibliography{hmc_optimal_tuning}

\end{document}